\newif\ifllncs
\llncsfalse

\ifllncs
    \documentclass{llncs}
    \usepackage[T1]{fontenc}
\usepackage[ascii]{inputenc}
\usepackage[bookmarksnumbered,linktocpage,hypertexnames=true,colorlinks=true,linkcolor=blue,urlcolor=blue,citecolor=blue,anchorcolor=green,breaklinks=true,pdfusetitle]{hyperref}
\usepackage{xcolor}
\usepackage{physics}
\usepackage{graphicx}
\usepackage{amsmath,amssymb}
\usepackage{mathtools}
\usepackage{bbm}
\usepackage{mleftright}
\usepackage[capitalize,nameinlink]{cleveref}
\usepackage{url,comment}
\usepackage{enumitem}
\usepackage{float,xspace}
\usepackage{textcomp}

\usepackage{amsthm}
\usepackage[normalem]{ulem}
\usepackage[english]{babel}

\newcommand{\E}{\mathbb{E}}
\else
\documentclass[12pt]{article}
\usepackage{styles}
\usepackage[affil-it]{authblk}
\fi
\usepackage{physics}

\ifllncs
\title{The Knowledge Complexity of Quantum Problems}
\else
\title{The Knowledge Complexity of Quantum Problems}
\date{}
\author{Giulio Malavolta}
\affil{Bocconi University}
\fi

\begin{document}
\maketitle

\begin{abstract}
    Foundational results in theoretical computer science have established that 
\emph{everything provable, is provable in zero knowledge}. However, this assertion fundamentally assumes a classical interpretation of computation and many interesting physical statements that one can hope to prove are not characterized. 

In this work, we consider decision problems, where the problem instance itself is specified by a (pure) quantum state. We discuss several motivating examples for this notion and, as our main technical result, we show that every quantum problem that is provable with an interactive protocol, is also provable in zero-knowledge. Our protocol achieves unconditional soundness and computational zero-knowledge, under standard assumptions in cryptography. 

In addition, we show how our techniques yield a protocol for the Uhlmann transformation problem that achieves a meaningful notion of zero-knowledge, also in the presence of a malicious verifier.

\end{abstract}

\section{Introduction}

Zero-knowledge proofs are a cornerstone of modern cryptography and complexity theory. It is known that \emph{everything provable is provable in zero knowledge} \cite{GMW,BGGSHKMR}, meaning that every statement that can be proven with an interactive protocol can also be proven without revealing any additional information to the verifier. However, this assertion fundamentally assumes a \emph{classical} model of computation: Statements are specified by bits and interactive protocols are executed by classical Turing machines. 

On the other hand, Nature operates by the laws of quantum mechanics. Properties about physical objects that one could hope to prove are not even well-defined in the classical model of computation. Take for instance the following simple task: You are given two bipartite pure states $\ket{\psi}_{AB}$ and $\ket{\phi}_{AB}$ and want to determine if they have the same reduced density on register $A$, i.e., if
\[
\mathrm{Tr}_B\left(\ketbra{\psi}{\psi}_{AB}\right) \stackrel{?}{=} \mathrm{Tr}_B\left(\ketbra{\phi}{\phi}_{AB}\right).
\]
Uhlmann's theorem \cite{UHLMANN} shows that a witness for this fact always exists, in the form of a unitary $\mathrm{U}$ mapping $(\mathrm{Id}_A\otimes \mathrm{U}_B)\ket{\psi}_{AB} = \ket{\phi}_{AB}$. Nevertheless, the current framework does not allow one to prove this fact in zero-knowledge.\footnote{We insist that the verifier does not have any further information about the states. If, for instance, the verifier knew a classical description of a quantum circuit generating such states, then the problem is solvable with known methods.}

To characterize what class of statements can be proven in zero-knowledge, clearly the first step is to characterize what can be proven at all. Indeed a series of recent works investigate complexity classes that are uniquely quantum. Rosenthal and Yuen \cite{RY21} studied the complexity for \emph{state synthesis} problems, characterizing the complexity of interactive protocol generating quantum states. Bostanci et al.~\cite{BEMPQY23} introduced complexity classes for \emph{unitary synthesis}, problems where a prover and a verifier interactively generate a unitary matrix. These can be thought of problems specified by classical inputs, but where the ``output'' is a quantum state. On the flip side, Kashefi and Alves \cite{KA04} and more recently Chia et al.~\cite{CCHS24} studied the complexity of decision problems (so the output is a classical bit) but where \emph{the problem instance itself} is specified by a quantum state.

All of the aforementioned works are focused on understanding the computational complexity of such problems, whereas the \emph{knowledge complexity} (i.e., what can be proven in zero-knowledge) is only considered for restricted cases, such for an honest verifier. At present, the question of what exactly is provable in zero-knowledge, according to its standard definition, remains wide open. The purpose of this work is to make progress on this question.

\subsection{Our Results}

In this work we mainly consider \emph{quantum promise} problems~\cite{KA04,CCHS24}, where problems instances $X = (X^{\text{yes}}, X^{\text{no}})$ are defined by two sets of quantum states. The complexity of interactive proofs for quantum promise problems is captured by the class $^*$QIP: A prover and a verifier are given as input (multiple copies of) a quantum state $\rho$, and the prover wants to convince the verifier that $\rho \in X^{\text{yes}}$. The verifier is required to be efficient, but the prover is potentially unbounded. Furthermore, it is required that the verifier rejects if $\rho \in X^{\text{no}}$.
If the states in the problem instance are restricted to \emph{pure states}, then we denote the complexity class by pQIP. Otherwise, if \emph{mixed states} are allowed, then we refer to the class as mQIP.

Our first contribution is a formal definition of the complexity classes p/mQIPzk where, in addition to the above properties, we require that a computationally-bounded verifier learns nothing beyond the fact that $\rho \in X^{\text{yes}}$. As standard in cryptography, this is formalized by the existence of an \emph{efficient} simulator whose messages are computationally indistinguishable from real protocols to the eyes of any, possibly \emph{malicious}, verifier. 

Our main technical result is a tight characterization of the class pQIPzk. We show that, under standard cryptographic assumptions, everything provable in pQIP, is also provable in zero-knowledge.

\ifllncs
\begin{theorem}[Main Theorem]
    If EFI pairs exist, then pQIP $\subseteq$ pQIPzk.
\end{theorem}
\else
\begin{theorem*}[Main Theorem]
    If EFI pairs exist, then pQIP $\subseteq$ pQIPzk.
\end{theorem*}
\fi
EFI pairs \cite{EFI} are families of quantum states that can be generated by a polynomial-size quantum circuit, which are (i) statistically far but (ii) computationally indistinguishable. The existence of EFI pairs is equivalent to the existence of quantum bit commitments, and there is strong evidence \cite{EFI} that they are the ``minimal'' assumption in quantum cryptography. For our theorem, assuming EFI pairs is tight, since it is shown in \cite{CCHS24} that EFI pairs are \emph{implied} by the average-case hardness of pQIPzk with an honest verifier. Furthermore, the computational assumption is only used to argue zero-knowledge, but soundness remains unconditional.

We remark that the algorithm for the honest prover in pQIP is not required to be efficient, so we cannot expect that the prover for our pQIPzk protocol will in general be efficient. However, for the case of pQMA (i.e., pQIP where the prover sends a single message) we can ask for a polynomial-time protocol, where the prover is non-deterministically given a quantum state as a \emph{witness}. The proof of our main theorem is an explicit compiler that achieves this strong efficiency guarantee, which makes it more suitable for cryptographic applications.

On the other hand, we show that when considering \emph{mixed states}, i.e., the complexity class mQMAzk, then such a protocol is unlikely to exist. Under standard cryptographic assumptions, we show that there exists a problem instance $X$ in mQMA such that any efficient protocol is either not sound, or not witness-indistinguishable (and therefore not zero-knowledge).

\ifllncs
\begin{theorem}[Limits on mQMAzk]
    If the learning with errors (LWE) problem is hard, then there exist problems in mQMA without an efficient zero-knowledge protocol.
\end{theorem}\else
\begin{theorem*}[Limits on mQMAzk]
    If the learning with errors (LWE) problem is hard, then there exist problems in mQMA without an efficient zero-knowledge protocol.
\end{theorem*}\fi
Finally, we outline how our techniques can be used to achieve a meaningful notion of zero-knowledge for the Uhlmann transformation problem \cite{BEMPQY23}, whereas prior works only considered honest verifiers.

\subsection{Related Work}

As already mentioned, prior work considering zero-knowledge in the context of quantum problems \cite{BEMPQY23,CCHS24} only studied this notion in the honest-verifier settings. Closest to our work is \cite{CCHS24}, where the authors systematically study (interactive and non-interactive) complexity classes where problem instances are quantum states, propose complete problems, and study separations across different classes. While they also discuss the notion of zero-knowledge for quantum promise problem, their results are in the honest-verifier settings and in fact it is explicitly left as an open problem to construct \emph{statistical} zero-knowledge protocols for quantum problems. Our work can be seen as a conditional resolution of this problem, albeit in the computational settings.

On the other side of the spectrum, zero-knowledge proofs for QMA \cite{ZKQMA} do consider malicious verifiers. The difference compared to our work is that instances in the complexity class QMA are classical bitstrings, despite the prover and the verifier being quantum algorithms. We refer the reader to \cite{CCHS24} for a detailed discussion. Of course, our theorem captures zero-knowledge for QMA as a special case.

\subsection{Motivation}

We discuss a few potential applications for zero-knowledge protocol for quantum problems.

\paragraph{Testing Properties of Quantum States.} Many areas of computer science and physics are concerned with learning properties of quantum states, see for example \cite{aaronson2007learnability,christandl2012reliable,pallister2018optimal,aaronson2018shadow,huang2020predicting,fawzi2024learning}. Furthermore, quantum property testing \cite{buhrman2008quantum,montanaro2013survey} studies how many copies of a given state are necessary in order to learn whether a state possesses a specific property. A natural question \cite{CCHS24} is whether interaction with an (untrusted) prover lowers the number of copies necessary. Using zero-knowledge protocols would ensure that even a malicious verifier can learn nothing other than the property holds.

In physics, one can imagine experimentalists wanting to certify that a given state is a ground state of some Hamiltonian, or that it lies in a particular phase of matter, without exposing the full state. This is different from certifying that \emph{there exists} a state that satisfies such property (which is a regular QMA statement), but rather we can prove that \emph{the state that we are currently holding} satisfies it.
Alternatively, one can prove that a given state has some hidden entanglement structure, without fully revealing the correlations.

\paragraph{Zero-Knowledge Uhlmann Beyond Honest Verifiers.} The Uhlmann transformation problem studies the complexity of implementing the unitary transformation from Uhlmann's theorem. A recent work of Bostanci et al.~\cite{BEMPQY23} identifies the Uhlmann transformation problem as a primitive of central interest, characterizing critical building blocks in quantum Shannon theory, quantum cryptography, quantum gravity, and quantum complexity. One of their main contributions is an \emph{honest-verifier} zero-knowledge protocol for the Uhlmann transformation problem, where at the end of the protocol the prover is guaranteed to apply the Uhlmann unitary to a state chosen by the verifier.

In this work, we show that it is possible to upgrade their protocol to achieve a meaningful notion of zero-knowledge also against a \emph{malicious} verifier: Loosely speaking, we show the existence of a simulator that simulates the real interaction with a \emph{single} oracle query to the Uhlmann unitary, capturing the fact that the only freedom of the verifier is which state to use, but nothing more.

\paragraph{Quantum Cryptography.} One of the textbook applications of zero-knowledge proofs for NP is to upgrade protocols secure against semi-honest adversaries to \emph{malicious security} \cite{GMW}. Naturally, one might expect that zero-knowledge for QMA might serve the same role for quantum protocols. Alas, the situation is more complicated, precisely because a quantum interaction consists of quantum states and proving the ``well-formedness'' of a quantum state is not a QMA statement (the statement itself involves a quantum state).

As an example, image the natural scenario where the verifier is given a commitment to a quantum circuit $\mathrm{Com}(C)$, together with (many copies of) a state $\ket{\psi}$ and a prover wants to convince them that $C\ket{0} = \ket{\psi}$. 
To prove such statements, one has to resort to ad-hoc techniques, whereas our framework provides a method to systematically approach this kind of problems.

\subsection{Open Problems}

We view zero-knowledge proofs for quantum problems as a promising avenue for future research, and many open problems remain. For example, all of our protocol require polynomially-many copies of the input state. Can we achieve meaningful soundness and zero-knowledge even when the verifier is given a \emph{single copy} of the state? Ideally we would like the test to be non-destructive.

Another fascinating question concerns the round complexity of our protocols. Our protocols have a (fixed) polynomial number of rounds in the security parameter, and using round-optimized secure computation protocols \cite{RoundMPC} one can achieve constant rounds, but under stronger cryptographic assumptions. Is it possible to construct constant-round protocols, assuming only the existence of EFI pairs?

Finally, a question left open by our work is whether a zero-knowledge protocol for mQMA or mQIP exists. Our theorem only rules out protocols with \emph{efficient provers}, but does not apply to the settings where the honest prover is not efficient. Although this would make the protocol less relevant for cryptographic applications, we would still like to know the answer to this question.

\subsection{Technical Outline}

For starters, let us consider the simpler settings of problem instances $X = (X^{\text{yes}}, X^{\text{no}})$ in pQMA. The prover and the verifier are given copies of a state $\ket{\psi}\in X^{\text{yes}}$ and at the end of the protocol the verifier should either accept or reject. An obvious first approach would be to let the prover and the verifier run a secure two-party computation protocol \cite{BCKM2PC,GLSV21}, where the prover additionally supplies the witness $\ket{\phi}$ as an input for the following functionality: Given $\ket{\phi}$ and $\ket{\psi}$, run the honest verification circuit and return its output. Alas, this idea quickly runs into trouble.

The problem is that the parties do not have a circuit to generate $\ket{\psi}$, so it can only be taken as part of the inputs of the secure computation. So the question is: Do we let the prover or the verifier supply $\ket{\psi}$? In the former case, the prover could supply a different state, for which he knows a witness, so one cannot argue soundness. For the latter, the verifier could input an arbitrary state, and use the output bit to learn non-trivial information about the prover's witness. Thus, zero-knowledge fails.

\paragraph{The ``De Finetti'' Approach.}
For the special case of pQMA, we can fix this problem by letting the prover and the verifier input many copies of the $\ket{\psi}$ state and securely compute the following functionality:
\begin{itemize}
    \item Step I: Randomly permute the copies of the prover and select a position $i^*$ among these copies.
    \item Step II: For all but the $i^*$-th copy, run a SWAP-test \cite{SWAP} with a different copy of the verifier. Abort if any of the tests fail.
    \item Step III: Run the honest verification circuit on input the $i^*$-th copy of $\ket{\psi}$ \emph{of the prover}, along with the witness $\ket{\phi}$. Return to the verifier the output of the circuit.
\end{itemize}
Note that the witness state $\ket{\phi}$ is only interacting with states supplied by the prover, which guarantees zero-knowledge -- the SWAP-tests are easy to simulate given enough copies of $\ket{\psi}$. To prove soundness against a malicious prover, we observe that the protocol (Step I) guarantees that the state $\rho$ in the registers of the prover is \emph{permutation-invariant}, so we can invoke the quantum De Finetti theorem \cite{RennerDeFinetti} to show that 
\[
\rho \approx \sum_ip_i \tau_i \otimes \dots \otimes \tau_i
\]
is approximated by a convex combination of separable states, where each register holds an identical state $\tau_i$. Crucially, we can interpret our test as a tensor product of classical-quantum channels, so we can appeal to a version of the De Finetti theorem \cite{BH13deFinetti} where the bound on the approximation scales favorably with the dimension of the Hilbert space. Given that the states in different registers are identical, we can use the fact that the SWAP-test guarantees that $\tau_i \approx \ketbra{\psi}{\psi}$ to derive the same conclusion on the $i^*$-th copy, although we never actually test it. Finally, soundness follows from the underlying pQMA verifier.

This approach works well for the pQMA case, but it stops working for pQIP instances. One issue is that the prover is not longer efficient, and it is therefore not clear how to reduce to the security of the secure computation protocol. The more severe problem is that the verification protocol is now interactive, so we can no longer claim that our verification is a product channel. Other versions of the De Finetti theorem, e.g., \cite{RennerDeFinetti,CHRISTANDLdeFINETTI}, do not yield meaningful bounds in our context. Instead, we resort to a \emph{completely different} approach, which also results into a protocol with better round complexity.

\paragraph{The ``Round Collapsing'' Approach.}
First, notice that in the \emph{honest-verifier} settings, the issue outlined above is not actually a problem: Since the verifier is honest, we trust it to supply the right state $\ket{\psi}$ as the input to the secure computation protocol. Clearly this is just deferring the problem (the main challenge was to achieve security against a malicious verifier), but it will soon become clear why this is a useful observation. A different issue is that the prover is no longer efficient, but this can be solved by running a new two-party computation protocol for each round of the verifier's computation, using quantum state commitments \cite{QSC} to ensure that the state of the verifier is kept consistent throughout different rounds. Due to the fact that the prover is inefficient, the security proof is delicate, by the hybrids can be made to work.

Next, we observe that the \emph{round-collapsing} theorems for QIP proofs, such as \cite{KW00,MW05,KKMV09}, result into protocols with a particularly useful structure. 
Let $\ket{\psi_v}$ be the initial state of the verifier (which consists of multiple copies of the statement $\ket{\psi}$). The $3$-message round-collapsed protocols have roughly of the following form.
\begin{enumerate}
    \item The verifier receives a register $W$ from the prover.
    \item The verifier samples a coin $b \in\{0,1\}$ uniformly at random and sends it to the prover.
    \item The prover returns a register $M$.
    \begin{itemize}
    \item If $b=0$: Apply $\mathrm{V}_2$ to $WM$ and measure the first qubit of $W$ in the computational basis. Accept if the outcome is $1$.
    \item If $b=1$: Apply $\mathrm{V}_1^\dagger$ to $WM$ and apply the projective measurement $\{\Pi_0, \Pi_1\}$ where
    \[
    \Pi_0 = \mathrm{Id} - \Pi_1 \quad \text{and} \quad \Pi_1 = \ketbra{\psi_v}{\psi_v}.
    \]
    Accept if the outcome $1$ is observed.
    \end{itemize}
\end{enumerate}
In particular, all the verifier is in this protocol is sending random challenge bits and project onto its initial state $\ket{\psi_v}$. An immediate difficulty is that we cannot quite do the latter, but we show that a SWAP-test provides a good enough approximation (with a slightly worse soundness bound), and the resulting protocol maintains soundness and honest-verifier zero-knowledge. 

We are omitting a lot of details here: Turns out that a straightforward application of existing theorems does not work. For instance, the recursive approach of \cite{KKMV09} would double the number of copies of $\ket{\psi}$ needed at every call, resulting in an exponential blow-up. In the actual protocol we need to apply a modified version of the Kitaev-Watrous \cite{KW00} compiler, followed by another modification of the Marriot-Watrous \cite{MW05} compiler, alternated with soundness amplifications steps. We refer to the technical sections for more details.

At this point we appear to be in good shape, since we have obtained an honest-verifier zero-knowledge protocol where the only thing that the verifier sends to the prover is random challenge bits, which seems to be innocuous even for a malicious verifier. Alas, we do not know how to prove the protocol secure against malicious verifier. For the informed reader, Watrous \cite{WatrousZK} rewinding strategy does not appear to be working without a good projector onto the initial state of the verifier.\footnote{This issue was already identified in \cite{CCHS24}, who left the problem open.} Instead, we resort to secure computation (again), but just to simulate the random coin tosses of the verifier. This can be shown to work, although once again the proof is quite delicate since the honest prover is not efficient, so care has to be taken when reducing to the security of the protocol. 

This completes the high-level description of the pQIPzk protocol. For more technical details, limitations for protocols with mixed problem instances, and for zero-knowledge for Uhlmann, we refer the reader to the technical sections.
\section{Preliminaries}

Throughout this work, we denote the security parameter by $\lambda$.
We denote by $1^{\lambda}$ the all-ones string of length~$\lambda$.
We say that a function~$f$ is negligible in the security parameter~$\lambda$ if~$f(\lambda) = {\lambda}^{-\omega(1)}$.

Unless stated otherwise, all random variables and probability distributions are finitely supported, that is, take values in finite sets. We write~$\mathrm{Id}$ for identity matrices or operators, and~$\Tr$ for the trace of a matrix or operator.
A unitary operator~$\mathrm{U}$ is one that satisfies~$\mathrm{U}\mathrm{U}^\dagger = \mathrm{U}^\dagger \mathrm{U}=\mathrm{Id}$, and a Hermitian operator~$\mathrm{U}$ is one such that~$\mathrm{U}^\dagger = \mathrm{U}$.

\subsection{Quantum Information}

In this section, we provide a brief overview of quantum information. For a more detailed introduction, we refer to \cite{NC00,Watrous2018theory}.
A \emph{(quantum) register}~$A$ consisting of $n$ qubits is associated with the Hilbert space $\mathcal{H}_A = (\mathbb{C}^2)^{\otimes n}$. Given two registers $A$ and $B$, we denote the composite register by~$AB$.
The corresponding Hilbert space is given by the tensor product $\mathcal{H}_{AB} = \mathcal{H}_A \otimes \mathcal{H}_B$.

The \emph{(quantum) state} of a register $A$ is described by a density operator $\rho_A$ on $\mathcal{H}_A$, which is a positive semi-definite Hermitian operator with trace equal to one.
A state is called \emph{pure} if it has rank one.
Thus, pure quantum states can be represented by unit vectors $\ket\psi_A \in \mathcal{H}_A$, with $\rho_A = \ketbra{\psi}{\psi}_A$.
For a quantum state~$\rho_{AB}$ on $\mathcal{H}_{AB}$, we denote $\rho_A = \Tr_B(\rho_{AB})\in \mathcal{H}_A$ the reduced state of $\rho_{AB}$ on $A$.

A \emph{(quantum) channel} $\mathrm{F}$ is a completely positive trace-preserving (CPTP) map from a register $A$ to a register $B$.
In other words, given any density matrix $\rho_A$, the channel $\mathcal{F}$ produces $\mathrm{F}(\rho_A) = \sigma_B$, which is another state on register~$B$, and the same applies when $\mathrm{F}$ is applied to the $A$-register of a quantum state $\rho_{AC}$, resulting in the quantum state~$\sigma_{BC} = (\mathrm{F} \otimes  \mathrm{Id})(\rho_{AC})$, where in this case $\mathrm{Id}$ denotes the identity channel.
For any unitary operator~$\mathrm{U}$, there is a quantum channel that maps any input state~$\rho$ to the output state $\mathrm{U} \rho \mathrm{U}^\dagger$.

A \emph{projective measurement} is defined by a set of projectors $\{ \Pi_j \}_j$ such that $\sum_j \Pi_j = \mathrm{Id}$.
A projector $\Pi$ is a Hermitian operator such that $\Pi^2 = \Pi$, that is, an orthogonal projection.
Given a state $\rho$, the measurement yields outcome $j$ with probability $p_j = \Tr(\Pi_j \rho)$, upon which the state changes to ${\Pi_j \rho \Pi_j/p_j}$.
A basis measurement is one where~$\Pi_j=\ketbra{e_j}{{e_j}}$ and the $\{\ket{e_j}\}_j$ (necessarily) form an orthonormal basis.

A \emph{positive operator-valued measure} (POVM) is a generalization of a projective measurement.
A POVM is defined by a set of positive semi-definite operators~$\{\mathrm{E}_j\}_j$ such that~$\sum_j \mathrm{E}_j = \mathrm{Id}$ (that is, the $\mathrm{E}_j$ no longer need to be projections).
As before, given a quantum state~$\rho$, the probability of obtaining outcome $j$ when performing the measurement is given by~$p_j = \Tr(\mathrm{E}_j \rho)$, but the state after the measurement is no longer uniquely specified.

The \emph{trace distance} between two states $\rho$ and $\sigma$, denoted by $\mathrm{Td}(\rho, \sigma)$, is defined as
\[
\mathrm{Td}(\rho, \sigma) = \frac{1}{2}\norm{\rho - \sigma}_1 = \frac{1}{2}\Tr\left(\sqrt{(\rho-\sigma)^\dagger(\rho - \sigma)}\right)
\]
The operational meaning of the trace distance is that $\frac{1}{2}(1 + \mathrm{Td}(\rho,\sigma))$ is the maximal probability that
two states $\rho$ and $\sigma$ can be distinguished by any (possibly unbounded) quantum channel or algorithm. That is,
\begin{align*}
    \mathrm{Td}(\rho, \sigma) = \max_{\mathcal A} \abs*{ \Pr\left( \mathcal A(\rho) = 1 \right) - \Pr\left( \mathcal A(\sigma) = 1 \right) }
\end{align*}
where the maximum is over arbitrary quantum channels~$\mathcal A$ that output a single bit. We recall the following version of the gentle measurement lemma.

\begin{lemma}[Gentle Measurement]\label{lmm:gentle}
    Let $\rho$ be a density matrix and $\Pi$ be an orthogonal projection. If $\Tr(\Pi\rho)\geq 1-\varepsilon$, then $\mathrm{Td}(\rho, \Pi \rho \Pi/\Tr(\Pi\rho)) \leq \sqrt{\varepsilon}$.
\end{lemma}

The \emph{fidelity} of two states $\rho$ and $\sigma$ can be defined as 
\begin{equation}\label{eq:fidelity}
F(\rho, \sigma) = \max_{\ket{\psi}, \ket{\phi}} \abs{\ip{\psi}{\phi}}^2
\end{equation}
where the maximum is taken over all purifications of $\rho$ and $\sigma$, i.e., all $K$, $\ket{\psi}$, and $\ket{\phi}$ such that
\[
\Tr_K(\ketbra{\psi}{\psi}) = \rho
\quad\text{and}\quad
\Tr_K(\ketbra{\phi}{\phi}) = \sigma.
\]
For all states $\rho$, $\sigma$, and $\tau$, the fidelity satisfies the following reverse triangle inequality
\begin{equation}\label{eq:fidtriangle}
    F(\rho,\sigma)^2 +F(\sigma,\tau)^2 \leq 1+F(\rho,\tau).
\end{equation}

We will also consider the notion of \emph{computational indistinguishability}.
Recall that a \emph{nonuniform QPT algorithm} $\mathcal A = \{\mathcal A_\lambda\}_\lambda$ consists of a family of quantum channels that can be implemented by polynomial-size quantum circuits that get quantum states of a polynomial number of qubits as advice.
We call $\mathcal A$ a \emph{nonuniform QPT distinguisher} if the channels output a single bit.

\begin{definition}[Computational Indistinguishability]\label{def:computational indist}
We say that two families of states~$\{\rho_\lambda\}_\lambda$ and $\{\sigma_\lambda\}_\lambda$ are \emph{computationally indistinguishable} if for every nonuniform QPT distinguisher~$\mathcal A = \{\mathcal A_\lambda\}_\lambda$ there exists a negligible function~$\mu$ such that the following holds for all~$\lambda$:
\begin{align}\label{eq:advantage}
     \abs{ \Pr\left( \mathcal A_\lambda(\rho_\lambda) = 1 \right) - \Pr\left( \mathcal A_\lambda(\sigma_\lambda) = 1 \right) } \leq \mu(\lambda).
\end{align}
We abbreviate this definition by writing $\{\rho_\lambda\}_\lambda \approx_c \{\sigma_\lambda\}_\lambda$.
\end{definition}

\subsection{The SWAP Test}

The $\mathrm{SWAP}$ operator is defined as the unitary operator mapping $\mathrm{SWAP}\ket{\psi}\ket{\phi} = \ket{\phi}\ket{\psi}$, for all pure states $\ket{\psi}$ and $\ket{\phi}$. The \emph{SWAP test} \cite{SWAP} is a well-known protocol to estimate the distance of two quantum states $\rho$ and $\sigma$. The test proceeds as follows:
\begin{itemize}
    \item Initialize an ancilla qubit in the $\ketbra{+}{+}$ state.
    \item Controlled on the ancilla qubit, apply a $\mathrm{SWAP}$ operator on the two input states.
    \item Measure the ancilla in the Hadamard basis and accept if the outcome is $0$.
\end{itemize}
It can be verified by direct calculation that, for the special case of one state $\sigma = \ketbra{\psi}{\psi}$ being pure, the output of the SWAP test is characterized by the POVM $\{ \mathrm{E}_0, \mathrm{E}_1\}$ applied to $\rho$, where
\begin{equation}\label{eq:swapPOVM}
    \mathrm{E}_0 = \frac{\mathrm{Id}+ \ketbra{\psi}{\psi}}{2} 
    \quad  \text{and}\quad \mathrm{E}_1 = \frac{\mathrm{Id}- \ketbra{\psi}{\psi}}{2}.
\end{equation}

\subsection{Quantum Message Authentication Codes}

We recall the notion of a quantum message authentication code (MAC) \cite{DNS10,BW16}. A quantum MAC consists of a collection of encoding and decoding channels
\[
\mathrm{Enc}^{(k)} : M \mapsto C \quad\text{and}\quad
\mathrm{Dec}^{(k)} : C \mapsto MF
\]
for all $k \in K$, where $K$ is the set of all possible keys. The register $F$ is referred to as the \emph{flag} register and it consists of a single qubit. For correctness, we require that for all states $\rho_M \in \mathcal{H}_M$ and all $k\in K_\lambda$ it holds that
\[
\mathrm{Dec}^{(k)}\mathrm{Enc}^{(k)}\rho_{M} = \rho_M \otimes\ketbra{1}{1}_F.
\]
In other words, the flag register is always set to $1$ (denoting accept) after the encoding and decoding channels are applied in sequence.

For security, we require that the encoded state cannot be modified without the decoding algorithm noticing. In particular, we require the strong notion of \emph{simulation security} as defined in \cite{BW16}. In the same work, the authors show that the Clifford code and the trap code satisfy this definition (unconditionally). Let us define, for a fixed $k$ and an adversarial unitary $\mathrm{A}$, the real channel as
\[
\mathrm{Real}^{(k)}(\mathrm{A}) : \rho_{MR} \mapsto  (\mathrm{Dec}^{(k)} \otimes \mathrm{Id})\mathrm{A}((\mathrm{Enc}^{(k)}\otimes \mathrm{Id})\rho_{MR})\mathrm{A}^\dagger.
\]
On the other hand, for two completely positive maps $(\mathrm{S}^{\mathrm{acc}},\mathrm{S}^{\mathrm{rej}})$ acting on the register $R$ and satisfying $\mathrm{S}^{\mathrm{acc}} +\mathrm{S}^{\mathrm{rej}} = \mathrm{Id}$, we define the ideal channel as
\[
\mathrm{Ideal}(\mathrm{S}^{\mathrm{acc}},\mathrm{S}^{\mathrm{rej}}) : \rho_{MR} \mapsto 
(\mathrm{Id}\otimes \mathrm{S}^{\mathrm{acc}})\rho_{MR} \otimes \ketbra{1}{1} + \Tr_M\left(\mathrm{Id} \otimes(\mathrm{S}^{\mathrm{rej}})\rho_{MR}\right) \otimes \tau_M \otimes \ketbra{0}{0}
\]
where $\tau_M$ is some fixed state. We are now ready to state formally the security of a quantum MAC.
\begin{definition}[Quantum MAC Security]\label{def:qMAC}
    A quantum MAC is $\varepsilon$-secure if for all unitaries $\mathrm{A}$ there exists a simulator $(\mathrm{S}^{\mathrm{acc}},\mathrm{S}^{\mathrm{rej}})$ such that for all states $\rho_{MR}$ we have
    \[
    \mathrm{Td}\left(
    \frac{1}{|K|}\sum_{k \in K}\mathrm{Real}^{(k)}(\mathrm{A})\rho_{MR}
    , \mathrm{Ideal}(\mathrm{S}^{\mathrm{acc}},\mathrm{S}^{\mathrm{rej}}) \rho_{MR}\right) \leq \varepsilon.
    \]
    Furthermore, if $\mathrm{Real}^{(k)}(\mathrm{A})$ is polynomial-size in the size of the input register $M$, then so is $\mathrm{Ideal}(\mathrm{S}^{\mathrm{acc}},\mathrm{S}^{\mathrm{rej}})$.
\end{definition}
In \cite{BW16} it is shown that the trap code with parameter $t$ is $\varepsilon$-secure for a function $\varepsilon=\varepsilon(t) = 3^{-(t+1)}$, so setting $t$ proportional to the security parameter $\lambda$, we can henceforth assume that the $\varepsilon$ is bounded by a negligible function in $\lambda$.

\subsection{Quantum Interactive Protocols}\label{sec:interactive}

An $m$-message quantum interactive protocol \cite{KW00} (for an odd $m$) is specified by an initial state $\ket{\psi_{\mathrm{init}}} \in \mathcal{H}_R \otimes  \mathcal{H}_W \otimes \mathcal{H}_M$ and two sequences of unitaries 
\[
\mathrm{P} = (\mathrm{P}_1, \dots, \mathrm{P}_{(m+1)/2})
\quad  \text{and}\quad
\mathrm{V} =(\mathrm{V}_1, \dots, \mathrm{V}_{(m+1)/2})
\]
where we refer the the former as the \emph{prover} and to the latter as the \emph{verifier}. Each unitary of the prover acts on registers $R M$  whereas each unitary of the verifier acts on $W M$. We refer to $R$ and $W$ as the workspace of the prover and verifier, respectively, and $M$ as the message register.

We assume without loss of generality that the prover sends the first message, then each round of the interaction proceeds as follows:
\begin{itemize}
    \item The prover applies $\mathrm{P}_i$ to $R  M$ and sends the message register $M$ to the verifier.
    \item The verifier applies $\mathrm{V}_i$ to $W M$ and sends $M$ back to the prover.
\end{itemize}
At the end of the interaction, the verifier applies the unitary $\mathrm{V}_{(m+1)/2}$ to $W M$ and applies the projective measurement $\Pi_\mathrm{Ver} =\{\Pi_0, \Pi_0\}$ where
\[
\Pi_0 = \ketbra{0}{0} \otimes \mathrm{Id}
\quad\text{and}\quad
\Pi_1 = \ketbra{1}{1} \otimes \mathrm{Id}
\]
and the projection acts non-trivially on the (without loss of generality) first qubit of $W$. If the result of the measurement is $1$, the verifier accepts. Otherwise, it rejects. Given this convention, we can uniquely specify the state of the system at round $k$ as follows:
\[
\mathrm{V}_k\mathrm{P}_k \cdots \mathrm{V}_1\mathrm{P}_1 \ket{\psi_\mathrm{init}}
\]
by appropriately tensoring the operators with identity.

In this work we consider \emph{families} of interactive protocols, where all of the variables above are parametrized by the security parameter $\lambda$, although for the sake of readability we often omit the index. We also require the verifier to be QPT, i.e., implementable by polynomial-size quantum circuits.

\subsection{Secure Computation}\label{sec:2pc}

Let $\mathrm{F} = \{\mathrm{F}_\lambda\}_\lambda$ be a family of quantum channels mapping $\mathrm{F}_\lambda : A_\lambda B_\lambda \mapsto \Tilde{A}_\lambda\Tilde{B}_\lambda$ implementable by a polynomial-size quantum circuit. A \emph{secure computation} protocol for $\mathrm{F}$ is an interactive protocol between two parties, Alice and Bob. Alongside their workspaces, Alice is given as input a state $\rho_{A_\lambda}$ and Bob is given as input a state $\rho_{B_\lambda}$ and both parties receive the security parameter $1^\lambda$. At the end of the interaction the \emph{output} of the protocol is defined to be the reduced state of their output registers $\Tilde{A}_\lambda\Tilde{B}_\lambda$. For the sake of readability we henceforth drop the dependency of $\lambda$, and we refer to the output of the protocol as 
\[
\mathrm{Alice}(\rho_A) \rightleftharpoons \mathrm{Bob}(\rho_B).
\]
We require the protocol to satisfy the following notion of security, that we refer to as malicious simulation security (it is also known as security with aborts).
\begin{definition}[Malicious Simulation Security]\label{def:2pc}
    Let $\mathrm{F}$ be a quantum channel implementable by a polynomial-size quantum circuit. An interactive protocol for computing $\mathrm{F}$ satisfies malicious simulation security against Alice, if the following holds. For any (malicious) QPT algorithm $\mathcal{A} = \{\mathcal{A}_\lambda\}_\lambda$, there exists a QPT simulator $\mathcal{S} = \{\mathcal{S}_\lambda\}_\lambda$, such that for any QPT distinguisher $\mathcal{D} = \{\mathcal{D}_\lambda\}_\lambda$, there exists a negligible function $\mu$ such that for all security parameters $\lambda\in\mathbb{N}$, all non-uniform bipartite advice state $\rho_{AD}$, and Bob's input $\rho_B$ it holds that:
    \[
    \abs{\Pr\left(\mathcal{D}_\lambda\left( \mathcal{A}_\lambda(\rho_A) \rightleftharpoons \mathrm{Bob}(\rho_B),\rho_D\right)= 1\right) -
    \Pr\left(\mathcal{D}_\lambda\left( \mathcal{S}_{\lambda,\mathrm{F}}(1^\lambda, \rho_A, \rho_B),\rho_D\right)= 1\right)
    } \leq \mu(\lambda)
    \]
    where the algorithm $\mathcal{S}_{\lambda,\mathrm{F}}$ is defined as follows:
    \begin{itemize}
        \item Run the simulator $\mathcal{S}_\lambda(1^\lambda, \rho_A)$, which outputs a quantum state $\rho_{A}^*$.
        \item Apply the channel $\mathrm{F}$ to $\rho_{AB}^*$ to compute an output $\rho_{\Tilde{A}\Tilde{B}}$.
        \item Continue executing the simulator on $\rho_{\Tilde{A}}$, which in the end outputs a certain state $\sigma$ and a bit $b\in\{0,1\}$.
        \item Output $(\sigma, \rho_{\Tilde{B}})$ if $b=0$ and  $(\sigma, \bot)$ otherwise.
    \end{itemize}
\end{definition}
Note that \cref{def:2pc} only refers to security against a corrupted Alice, but we can define security against a corrupted Bob analogously, swapping the roles of the parties. We say that the protocol satisfies \emph{statistical} security if the same definition holds, except that the adversary and the distinguisher (and consequently the simulator) are not required to be QPT, i.e., security holds against \emph{unbounded} adversaries. It is well-known that a secure computation protocol cannot satisfy statistical security against both a corrupted Alice and a corrupted Bob, but there exist protocols that satisfy statistical security against a corrupted Bob and standard (computational) security against Alice. In this case, we say that a protocol satisfies \emph{one-sided statistical security}.

In \cite{EFI} it is show, using techniques from \cite{BCKM2PC,AQY22}, that one-sided statistical secure computation for \emph{classical functions} exist assuming the existence of EFI pairs. EFI pairs generalize the notion of one-way functions and consists of two families of quantum states $\{\rho_\lambda\}_\lambda$ and $\{\sigma_\lambda\}_\lambda$ that are (i) efficiently preparable, i.e., they can be generated by a polynomial-size quantum circuit, (ii) statistically far, and (iii) computationally indistinguishable, that is
\[
\{\rho_\lambda\}_\lambda \approx_c \{\sigma_\lambda\}_\lambda.
\]
On the other hand, \cite{DNS10} shows that \emph{two-sided statistically} secure computation for any quantum channel is possible, assuming an ideal implementation of a secure computation for classical functionality. In \cite{DNS10} only one party receives the output. Thus, combining all of the aforementioned works, we obtain the following implication (already observed in \cite{EFI}).
\begin{lemma}[\cite{EFI,DNS10}]\label{lmm:2pcEFIoneside}
    If EFI pairs exist, then one-sided statistical secure computation exists for all quantum channels implementable by a polynomial-size quantum circuit, where only one party receives the output.
\end{lemma}
Furthermore, we can easily show that a one-sided secure computation where one party receives the output, implies one where both parties receive the output. Note that this does not contradict known impossibility results on fairness, since we only achieve security with aborts. The proof is a straightforward consequence of the existence of quantum MACs.
\begin{lemma}[EFIs Imply Secure Computation]\label{lmm:2pcEFI}
    If EFI pairs exist, then one-sided statistical secure computation exists for all quantum channels implementable by a polynomial-size quantum circuit.
\end{lemma}
\begin{proof}
    Let $\mathrm{F}: AB \mapsto \Tilde{A}\Tilde{B}$ be the quantum channel that we wish to implement and, for a key $k$, define $\mathrm{F}^{(k)}:AB \mapsto \Tilde{A}\Tilde{B}$ to be the channel that computes $\mathrm{F}$ and applies the encoding algorithm $\mathrm{Enc}^{(k)}$ of the quantum MAC to $\Tilde{B}$. Then define $\Tilde{\mathrm{F}}: ABK \mapsto \Tilde{A}\Tilde{B}$ as the controlled application of $\mathrm{F}^{(k)}$ on the $K$ register. Observe that if $\mathrm{F}$ is implementable by a polynomial-size quantum circuit, then so is $\Tilde{\mathrm{F}}$.
    The protocol proceeds as follows:
    \begin{itemize}
        \item Alice provides her input $\rho_{A}$.
        \item Bob provides his input $\rho_{B}$ along with a uniformly sampled key $k$ on register $K$.
        \item Alice and Bob run the secure computation for $\Tilde{\mathrm{F}}$ (whose existence is guaranteed by \cref{lmm:2pcEFIoneside}), where Alice receives the output.
        \item Alice sends the register $\Tilde{B}$ to Bob.
        \item Bob runs the decoding algorithm $\mathrm{Dec}^{(k)}$ of the quantum MAC and returns the resulting state if the verification passes. Otherwise it outputs $\bot$.
    \end{itemize}
    Security against a corrupted Bob follows with a straightforward reduction against the simulation security of the secure computation protocol. On the other hand, for a corrupted Alice, we can first consider a hybrid experiment where the output of the secure computation protocol is simulated by the ideal functionality. Indistinguishability from the real experiment follows by \cref{def:2pc}. In this hybrid experiment, the encoding algorithm is honestly computed and the key $k$ is uniformly sampled. Thus, we can invoke the simulator from the security of the quantum MAC (see \cref{def:qMAC}) to simulate the output of Alice acting as the identity on register $\Tilde{B}$, conditioned on Bob accepting. Indistinguishability of the simulation follows by a triangle inequality.
\end{proof}

\subsection{Quantum State Commitments}\label{sec:qscom}

A quantum state commitment scheme \cite{QSC} consists of a unitary $\mathrm{Com}$ acting on two registers $M$ and $W$ with consisting of $n$ and $\lambda$ qubits, respectively. Committing to a state $\rho_M$ is done by computing
\[
\mathrm{Com}\left(\rho_M \otimes \ketbra*{0^\lambda}{0^\lambda}_W\right) \in \mathcal{H}_C \otimes \mathcal{H}_D
\]
and by dividing the resulting state into two registers $C$ and $D$, which we refer to as the commitment register and the decommitment register, respectively. Verification of a commitment is done via the canonical procedure: Apply $\mathrm{Com}^\dagger$ to $CD$ and apply the projective measurement $\Pi_{\mathrm{Com}} = \{\Pi_0, \Pi_1\}$, where
\[
\Pi_0 = \mathrm{Id} - \Pi_1
\quad\text{and}\quad
\Pi_1 = \ketbra*{0^\lambda}{0^\lambda}_W \otimes\mathrm{Id}_M.
\]
The verification accepts if the outcome $1$ is observed. Security of a commitment requires the notion of hiding, meaning the the state $\rho_M$ is hidden by the commitment, and binding, meaning that one cannot open to a different state acting only on the decommitment register $D$. For our purposes, we will need the following form of binding, which we define next.

\begin{definition}[Double Opening Binding]\label{def:swapbind}
Consider the experiment $\mathrm{ExpDoubleOpen}$, played between an adversary $\mathcal{A} = \{\mathcal{A}_\lambda\}_\lambda$ and a challenger:
\begin{itemize}
    \item The adversary sends $CD$.
    \item The challenger checks that the decommitment is valid by applying $\mathrm{Com}^\dagger$ and applying the measurement $\Pi_{\mathrm{Com}}$, aborting if the first outcome is observed.
    \item The challenger samples a bit $b \gets \{0, 1\}$ and does one the following:
    \begin{itemize}
        \item If $b = 0$, it applies $\mathrm{Com}$ to recompute the commitment/decommitment.
        \item If $b = 1$, it applies the $\mathrm{SWAP}$ operator to $MM'$, where $M'$ is a register initialized in state $\ketbra{0^n}{0^n}$. It then applies $\mathrm{Com}$ to recompute the commitment/decommitment.
    \end{itemize}
    Finally, it sends $D$ back to the adversary (but not $C$).
    \item The adversary sends another decommitment on $D$ back to the challenger.
    \item The challenger checks that the decommitment is valid by applying $\mathrm{Com}^\dagger$ and applying the measurement $\Pi_{\mathrm{Com}}$, aborting if the first outcome is observed.
    \item The challenger then does one the following:
    \begin{itemize}
        \item If $b = 0$, it applies the $\mathrm{SWAP}$ operator to $MM'$, where $M'$ is a register initialized in state $\ketbra{0^n}{0^n}$. It then applies $\mathrm{Com}$ to recompute the commitment/decommitment.
        \item If $b = 1$, it applies $\mathrm{Com}$ to recompute the commitment/decommitment.
    \end{itemize}
    Finally, it sends all of its registers (including $M'$) back to the adversary.
    \item The adversary outputs a guess $b'$ and wins if $b' = b$.
\end{itemize}
We say that a commitment scheme satisfies double-opening binding if for all nonuniform QPT attackers $\mathcal{A} = \{\mathcal{A}_\lambda\}_\lambda$, there exists a negligible function $\mu$ such that for all security parameters $\lambda$ it holds that
\[
\Pr(\mathrm{ExpDoubleOpen}(1^\lambda, \mathcal{A}_\lambda) =1) \leq \frac{1}{2} + \mu(\lambda).
\]
\end{definition}
Likewise, we say that a commitment scheme satisfies \emph{statistical} double opening binding if the above holds for all (possibly unbounded) attackers. 

It is shown in \cite{QSC} that the above definition is equivalent to the notion of \emph{SWAP binding}. Furthermore, the authors show that quantum bit commitments (which are equivalent to EFI pairs) suffice to construct commitments that are statistically hiding and computationally SWAP binding.

In addition, in \cite{QSC} they prove a duality theorem showing that if a commitment scheme is statistically hiding and computationally SWAP binding, then the same scheme can be turned into a computationally hiding and statistically SWAP binding simply by inverting the roles of the registers $C$ and $D$. Combining these implications, we obtain that EFI pairs imply the existence of a commitment scheme that simultaneously satisfies (i) computational double opening binding and (ii) statistical double opening binding, by switching the roles of the registers $C$ and $D$. This is the only property that we are going to need in our work, so we refrain from formally defining hiding.

\section{Definitions}\label{sec:def}

In this section we define quantum problems and zero-knowledge proofs, which is the main focus of this work. We begin by defining problem instances.

\begin{definition}[Quantum Decision Problems]
    A quantum promise problem $X = (X^{\text{yes}}, X^{\text{no}})$ is defined by a pair of sets of quantum states $X^{\text{yes}}$ and $X^{\text{no}}$. A state $\rho$ is called a yes instance if $\rho \in X^{\text{yes}}$ and a no instance if $\rho \in X^{\text{no}}$. 
\end{definition}
If the sets consist exclusively of pure states, then we refer to such problems as \emph{pure} quantum decision problems. To disambiguate, we otherwise call them \emph{mixed} quantum decision problems. One can alternatively state state problems as \emph{pairs} of states and bits (where the bit denotes whether it is a yes instance or a no instance), but we stick to the canonical for quantum promise problem.

Fix a universal set of gates. We say that a family of quantum circuits $\{V_\lambda\}_\lambda$ is P-uniform if there exists a polynomial-time Turing Machine $\mathrm{M}$ that, on input $1^\lambda$ outputs a classical description (e.g., as a sequence of gates) of $V_\lambda$. The definition is naturally extended to interactive verification circuits (as defined in \cref{sec:interactive}). In a slight abuse of notation, we ignore the ancilla registers provided as an input to the verifier (in part because they can be specified as part of the input state), but it should be understood that the verifier's workspace is initialized to the $\ketbra*{0 \dots 0}{0 \dots 0}$ state.

We are now ready to define the complexity class of decision quantum problems. Our definition is taken almost in verbatim from \cite{CCHS24}.

\begin{definition}[$^*$QIP]\label{def:QIP}
    A collection of promise problems $X = (X^{\text{yes}}, X^{\text{no}})$ is in $^*$QIP if there exists polynomial $m, p, v$ and a P-uniform verifier interactive circuit family $\mathrm{V}=\{\mathrm{V}_\lambda\}_\lambda$ such for all sufficiently large $\lambda$ and all states $\rho$ the following holds.
    \begin{itemize}
        \item (Completeness) If $\rho \in X^{\text{yes}}$, then there exists an $m(\lambda)$-message prover $\mathrm{P}_\lambda$ such that
        \[
        \Pr\left(\mathrm{P}(\rho^{\otimes p(\lambda)}) \rightleftharpoons \mathrm{V}_\lambda(\rho^{\otimes v(\lambda)}) = 1 \right) \geq 2/3.
        \]
        \item (Soundness) If $\rho \in X^{\text{no}}$, then for all $m(\lambda)$-message provers $\mathcal{P}_\lambda$ we have
        \[
        \Pr\left(\mathcal{P}_\lambda(\rho^{\otimes p(\lambda)}) \rightleftharpoons \mathrm{V}_\lambda(\rho^{\otimes v(\lambda)}) = 1 \right) \leq 1/3.
        \]
    \end{itemize}
\end{definition}
Following the notation from \cite{CCHS24}, we can specialize the above definition for the case of pure quantum decision problems (which we refer to as pQIP) by requiring that the problem instance consists exclusively of pure states. On the other hand, we call the class mQIP if the states can be mixed. Note that the class pQIP (mQIP, respectively) where $m=1$ corresponds to the class pQMA (mQMA, respectively) and by requiring that the message is classical we obtain the classes pQCMA and mQCMA.

We remark that our choice of constants for the completeness and soundness bound is arbitrary. This is justified by the fact that \cite{CCHS24} proved an amplification theorem showing that, as long as the completeness-soundness gap is at least $1/\mathrm{poly}(\lambda)$, then completeness and soundness can be amplified to be \emph{exponentially close} to $1$ and $0$, respectively. Therefore, we can henceforth assume that this holds without loss of generality. Finally, we also note that one could require the stronger definition of soundness where the prover is given an infinite amount of copies of the state, but for simplicity we stick to above definition although everything in this work generalizes to the stronger settings.

Next, we define the class of problems with a zero-knowledge proof, directly for the case of negligible completeness and soundness error. Same as above, we refer to pQIPzk if the states are required to be pure and otherwise mQIPzk.

\begin{definition}[$^*$QIPzk]\label{def:QIPzk}
    A collection of promise problems $X = (X^{\text{yes}}, X^{\text{no}})$ is in $^*$QIPzk if there exists polynomial $m, p, v,s$ and a P-uniform verifier interactive circuit family $\mathrm{V} = \{\mathrm{V}_\lambda\}_\lambda$ such that the following holds.
    \begin{itemize}
        \item (Completeness) For all $\rho \in X^{\text{yes}}$ there exists a prover interactive circuit family $\mathrm{P} = \{\mathrm{P}_\lambda\}_\lambda$ and a negligible function $\mu$ such that 
        for all $\lambda\in\mathbb{N}$ we have
        \[
        \Pr\left(\mathrm{P}_\lambda(\rho^{\otimes p(\lambda)}) \rightleftharpoons \mathrm{V}_\lambda(\rho^{\otimes v(\lambda)}) = 1 \right) \geq 1 - \mu(\lambda)
        \]
        where the interaction consists of $m(\lambda)$ messages.
        \item (Soundness) For all $\rho \in X^{\text{no}}$ and all (malicious) provers $\mathcal{P} = \{\mathcal{P}_\lambda\}_\lambda$ there exists a negligible function $\mu$ such that 
        for all $\lambda\in\mathbb{N}$ we have
        \[
        \Pr\left(\mathcal{P}_\lambda(\rho^{\otimes p(\lambda)}) \rightleftharpoons \mathrm{V}_\lambda(\rho^{\otimes v(\lambda)}) = 1 \right) \leq  \mu(\lambda)
        \]
        where the interaction consists of $m(\lambda)$ messages.
        \item (Zero Knowledge) For all $\rho \in X^{\text{yes}}$ and all (malicious) QPT algorithms $\mathcal{A} = \{\mathcal{A}_\lambda\}_\lambda$, there exists a QPT simulator $\mathcal{S} = \{\mathcal{S}_\lambda\}_\lambda$, such that for any QPT distinguisher $\mathcal{D} = \{\mathcal{D}_\lambda\}_\lambda$, there exists a negligible function $\mu$ such that for all security parameters $\lambda\in\mathbb{N}$ and all non-uniform bipartite advice states $\rho_{AD}$ we have
    \[
    \abs{\Pr\left(\mathcal{D}_\lambda\left( \mathrm{P}_\lambda(\rho^{\otimes p(\lambda)}) \rightleftharpoons \mathcal{A}_\lambda(\rho_A),\rho_D\right)= 1\right) -
    \Pr\left(\mathcal{D}_\lambda\left( \mathcal{S}_\lambda(\rho_A, \rho^{\otimes s(\lambda)}),\rho_D\right)= 1\right)
    } \leq \mu(\lambda)
    \]
where the interaction consists of $m(\lambda)$ messages.
    \end{itemize}
\end{definition}
Let us comment that what prevents the simulator to run the honest prover (and thus makes the definition non-trivial) is the fact that we impose the requirement that the simulator is QPT, whereas the honest prover is potentially unbounded.

In contrast to \cref{def:QIP}, the above definition does not immediately specialize to pQMA or mQMA statements because, although the (standard) verification protocol for p/mQMA is non-interactive, the zero-knowledge protocol could be interactive. In the latter case, we additionally need to make sure that the honest prover is QPT, when given copies of the witness state. Besides syntactical modifications, this only changes the completeness requirement of \cref{def:QIPzk}. 
\ifllncs
\else
In the following definition we refer to the single message from the prover to the verifier as the \emph{witness}.
\begin{definition}[$^*$QMAzk]\label{def:QMAzk}
    A collection of promise problems $X = (X^{\text{yes}}, X^{\text{no}})$ is in $^*$QMAzk if there exists polynomial $m, p, v,s$ and a P-uniform prover/verifier interactive circuit family $\mathrm{P} = \{\mathrm{P}_\lambda\}_\lambda$ and $\mathrm{V} = \{\mathrm{V}_\lambda\}_\lambda$ such that the following holds.
    \begin{itemize}
        \item (Completeness) For all $\rho \in X^{\text{yes}}$ and $\sigma$ the corresponding witness there exits a negligible function $\mu$ such that 
        for all $\lambda\in\mathbb{N}$ we have
        \[
        \Pr\left(\mathrm{P}_\lambda(\sigma^{\otimes p(\lambda)},\rho^{\otimes p(\lambda)}) \rightleftharpoons \mathrm{V}_\lambda(\rho^{\otimes v(\lambda)}) = 1 \right) \geq 1 - \mu(\lambda)
        \]
        where the interaction consists of $m(\lambda)$ messages.
        \item (Soundness) For all $\rho \in X^{\text{no}}$ and all (malicious) provers $\mathcal{P} = \{\mathcal{P}_\lambda\}_\lambda$ there exists a negligible function $\mu$ such that 
        for all $\lambda\in\mathbb{N}$ we have
        \[
        \Pr\left(\mathcal{P}_\lambda(\rho^{\otimes p(\lambda)}) \rightleftharpoons \mathrm{V}_\lambda(\rho^{\otimes v(\lambda)}) = 1 \right) \leq  \mu(\lambda)
        \]
        where the interaction consists of $m(\lambda)$ messages.
        \item (Zero Knowledge) For all $\rho \in X^{\text{yes}}$ and all (malicious) QPT algorithms $\mathcal{A} = \{\mathcal{A}_\lambda\}_\lambda$, there exists a QPT simulator $\mathcal{S} = \{\mathcal{S}_\lambda\}_\lambda$, such that for any QPT distinguisher $\mathcal{D} = \{\mathcal{D}_\lambda\}_\lambda$, there exists a negligible function $\mu$ such that for all security parameters $\lambda\in\mathbb{N}$ and all non-uniform bipartite advice states $\rho_{AD}$ we have
    \[
    \abs{\Pr\left(\mathcal{D}_\lambda\left( \mathrm{P}_\lambda(\sigma^{\otimes p(\lambda)},\rho^{\otimes p(\lambda)}) \rightleftharpoons \mathcal{A}_\lambda(\rho_A),\rho_D\right)= 1\right) -
    \Pr\left(\mathcal{D}_\lambda\left( \mathcal{S}_\lambda(\rho_A, \rho^{\otimes s(\lambda)}),\rho_D\right)= 1\right)
    } \leq \mu(\lambda)
    \]
    holds for any valid witness $\sigma$ and the interaction consists of $m(\lambda)$ messages.
    \end{itemize}
\end{definition}
\fi
\ifllncs
\else

\section{Zero Knowledge for p/mQMA}

\subsection{Zero Knowledge Proofs for pQMA}

We present a simple protocol for proving statements in pQMA, assuming the existence of a one-sided statistical secure computation protocol. Let $X = (X^{\text{yes}}, X^{\text{no}})$ be a pQMA problem instance, and let us assume without loss of generality (see \cref{sec:def}) that the protocol has completeness $1- \mu(\lambda)$ and soundness $\mu(\lambda)$, for some negligible function $\mu$. Our protocols is parametrized by two polynomials $p = p(\lambda)$ and $q = q(\lambda)$, whose exact values we will specify later. For a (pure) state $\ket{\psi} \in (\mathbb{C}^{2})^{\otimes n}$ and a security parameter $\lambda$, the protocol consists of the joint secure computation of the following functionality, with statistical security against a corrupted prover.

\begin{itemize}
    \item (Prover Input) The prover inputs $p$ copies of the state $\ketbra{\psi}{\psi}^{\otimes p}$ and of the witness $\sigma^{\otimes p}$.
    \item (Verifier Input) The verifier inputs $q$ copies of the state $\ketbra{\psi}{\psi}^{\otimes q}$.
    \item (Computation) The functionality proceeds as follows:
    \begin{itemize}
        \item Sample a random subset $S$ of $\{1, \dots, p\}$ with $\abs{S} = q$ and a random element of $s^* \in \{1, \dots, p\} \setminus S$. 
        \item For all $s \in S$, run a SWAP-test between the $s$-th copy of the input state of the prover and a fresh copy of the verifier. Abort if any of the tests fail.
        \item Apply the projection $\mathrm{V}^\dagger (\ketbra{1}{1}\otimes \mathrm{Id}) \mathrm{V}$ to the joint state of the $s^*$-th copy of $\sigma$ and $\ketbra{\psi}{\psi}$ of the prover and accept if the projection succeeds.
    \end{itemize}
\end{itemize}
Completeness of the protocol is immediate.
In order to prove soundness, we need to recall a technical lemma. Given a POVM $\{\mathrm{E}_j\}_j$ we associate a measurement map 
\[
\Lambda(\rho) = \sum_j \Tr(\mathrm{E}_j \rho) \ketbra{e_j}{e_j}
\]
with an orthonormal basis $\{\ket{e_j}\}_j$. We refer to this map as a quantum-classical channel. The following theorem (known as the dimension-independent quantum De Finetti) establishes a bound on the distance between a permutation-invariant state and a product state.
\begin{lemma}[\cite{BH13deFinetti}]\label{thm:definetti}
    Let $\rho_{A_1 \dots A_p}$ be a permutation-invariant state on registers 
    $A_1 \dots A_p$ where each register contains $n$ qubits. For any 
    $0 \leq q \leq p$ there exists $\{\tau_i\}_i$ and  $\{p_i\}_i$ such that
    \[
    \max_{\Lambda_1 \dots \Lambda_q} \norm{(\Lambda_1 \otimes \dots \otimes \Lambda_q)\left(\rho_{A_1 \dots A_p} - \sum_i p_i \tau_i \otimes \dots \otimes \tau_i\right)}_1 \leq \sqrt{\frac{2q^2 n}{p-q}}
    \]
    where $\Lambda_i$ are quantum-classical channels.
\end{lemma}
We are now ready to prove a bound on the soundness. Note that the bound claimed below consists of three summands and, for a large enough $p$ and $q$, the first two summands are vanishing. Thus we obtain an overall constant bound, and soundness can be amplified by sequential repetition.
\begin{lemma}[Soundness]
    The protocol as described above is sound with a maximal success probability negligibly close to
    \[
    \sqrt{\frac{2q^2 n}{p-q}} + 0.99^q + \frac{1}{\sqrt{50}}.
    \]
\end{lemma}
\begin{proof}
We analyze the protocol in a setting where the adversarial prover interacts with the ideal functionality, and the bound on the soundness directly follows from (statistical) indistinguishability of the secure computation protocol. 

Let $A_i$ be the registers containing the input states of the prover, let $B_i$ the registers containing the witness state, and let $\rho_{A_1B_1 \dots A_pB_p}$ be the reduced density of these registers. Sampling a random subset of the registers is equivalent to randomly permuting all pairs and selecting the first $q + 1$ of them. Thus the state of the system can be described by
\[
\Tilde{\rho} = \frac{1}{|S_p|} \sum_{\mathrm{P}\in S_p} \mathrm{P} \rho \mathrm{P}^\dagger
\]
which is easily seen to be permutation-invariant. Both the SWAP-test and the last test performed by the functionality are classical-quantum channels and furthermore they are all in tensor product with each other. Thus, by \cref{thm:definetti} we can approximate their output by a classical mix over identical copies of some state $\tau_i$, up to an additive factor. To bound the soundness of the protocol, we consider the index $i$ such that $\tau_i$ has the highest success probability, and the bound follows by convexity.

For notational convenience, we rename $\tau_i = \tau^*$ and, for all $j$, we let $\tau_{A_j}^* = \Tr_{B_j}(\tau^*_{A_jB_j})$. Let $\mathrm{E}_0$ be the positive semi-definite operator of the SWAP-test for a pure state $\ketbra{\psi}{\psi}$, as defined in \cref{eq:swapPOVM}.
For all indices $j$, we consider two cases:
\begin{itemize}
    \item $\mathrm{Tr}(\mathrm{E}_0 \tau_{A_j}^*) < 1 -\frac{1}{100}$: Since all the SWAP-tests are executed with independent copies of $\ketbra{\psi}{\psi}$ and the $\tau^*_{A_jB_j}$ are all in tensor product, we can bound the probability that all tests accept by $0.99^q$.
    \item $\mathrm{Tr}(\mathrm{E}_0 \tau_{A_j}^*) \geq 1 -\frac{1}{100}$: In this case, we have that the SWAP test succeeds with probability close to $1$. More precisely, we can write
    \[
    \Tr(\mathrm{E}_0 \tau_{A_j}^*) = \Tr\left( \frac{\mathrm{Id}+\ketbra{\psi}{\psi}}{2}\tau_{A_j}^*\right) = \frac{1}{2} \Tr\left( \tau_{A_j}^*\right) + \frac{1}{2}\Tr\left(\ketbra{\psi}{\psi} \tau_{A_j}^*\right) \geq 1 -\frac{1}{100}
    \]
    since $\mathrm{Tr}( \tau_{A_j}^*) \leq 1$ we can conclude that
    \[
    \Tr\left(\left(\ketbra{\psi}{\psi}\otimes \mathrm{Id}\right) \tau_{A_jB_j}^*\right) = \Tr\left(\ketbra{\psi}{\psi} \tau_{A_j}^*\right) \geq 1 - \frac{1}{50}.
    \]
    Then, by \cref{lmm:gentle} we can bound
    \[
    \mathrm{Td}\left(\tau^*_{A_jB_j}, \frac{\left(\ketbra{\psi}{\psi}\otimes \mathrm{Id}\right)\tau_{A_jB_j}^*\left(\ketbra{\psi}{\psi}\otimes \mathrm{Id}\right)}{\Tr\left(\left(\ketbra{\psi}{\psi}\otimes \mathrm{Id}\right) \tau_{A_jB_j}^*\right)}\right) \leq \frac{1}{\sqrt{50}}
    \]
    and we can bound with a triangle inequality the probability that the projection $\mathrm{V}^\dagger (\ketbra{1}{1}\otimes \mathrm{Id}) \mathrm{V}$ succeeds to be negligibly close to $1/\sqrt{50}$, by observing that the state on the RHS corresponds to a correct initial state of the verifier.
\end{itemize}
By a union bound, we obtain that the success probability of the adversary is bounded by
\[
\sqrt{\frac{2q^2 n}{p-q}} + 0.99^q + \frac{1}{\sqrt{50}}
\]
plus a negligible factor, as desired.
\end{proof}

Finally, we show that the protocol satisfies zero-knowledge.
\begin{lemma}[Zero-Knowledge]
    If EFI pairs exist, then the protocol as described above satisfies zero-knowledge.
\end{lemma}
\begin{proof}
    The simulator simply runs the simulator of the secure computation protocol to extract the input state of the (malicious) verifier. Then it uses copies of the state $\ketbra{\psi}{\psi}$ to run the SWAP-tests. If any of the tests fail, the output is programmed to be $0$, otherwise the output is set to $1$. Clearly the simulator is QPT, and the view of the adversary is indistinguishable by \cref{def:2pc}.
\end{proof}
Overall, we have shown the following theorem.
\begin{theorem}[EFIs Imply pQMAzk]
    If EFI pairs exist, then pQMA $\subseteq$ pQMAzk.
\end{theorem}

\subsection{Evidence Against Zero-Knowledge for Mixed States}

Next, we show that zero-knowledge proofs for mQMA with an efficient prover are unlikely to exist. Towards this we consider a public key encryption scheme $(\mathrm{Gen},\mathrm{Enc}, \mathrm{Dec})$ encrypting a single bit with perfect correctness that satisfies the following properties.
\begin{itemize}
    \item (Commitment) For a randomly sampled pair $(\mathrm{pk,sk}) \in \mathrm{Gen}(1^\lambda)$, we have that $\mathrm{pk}$ is a statistically binding commitment to $\mathrm{sk}$.
    \item (CPA Security) For a randomly sampled $(\mathrm{pk,sk}) \in \mathrm{Gen}(1^\lambda)$ we have that
    \[
    \left\{\mathrm{pk},\mathrm{Enc}(\mathrm{pk}, 0)\right\}_\lambda \approx_c \left\{\mathrm{pk},\mathrm{Enc}(\mathrm{pk}, 1)\right\}_\lambda.
    \]
    By a standard hybrid argument, indistinguishability extends to a polynomial number of ciphertexts in the view of the distinguisher.
\end{itemize}
It is well-known that CPA-secure encryption can be constructed assuming the hardness of the learning with errors (LWE) problem \cite{Regev09}. For the commitment property, it is possible to make \emph{any} encryption scheme satisfying this property, simply by adding a statistically binding and computationally hiding commitment to the public key. The new public key is defined to be $(\mathrm{pk}, \mathrm{Com(sk)})$ and the new secret key is $(\mathrm{sk}, s)$ where $s$ is the opening of the commitment. Security of the augmented scheme follows again by a standard argument, where we first replace the commitment with $\mathrm{Com}(0)$ and then we appeal to the CPA-security of the scheme.
Statistically binding commitments can be constructed from any one-way function \cite{NaorCom}, and in particular from LWE. Thus we can conclude that, assuming the hardness of LWE there exists a scheme satisfying all of the above properties. 

A direct consequence of CPA security is that, for a randomly sampled $\mathrm{pk}$ in the support of $\mathrm{Gen}(1^\lambda)$, the following two mixed states are computationally indistinguishable
\begin{equation}\label{eq:CPA}
\left\{\sum_{r}f(r)\ket{\mathrm{pk}, \mathrm{Enc}(\mathrm{pk}, 0; r)} \right\}_\lambda\approx_c
\left\{\sum_{r}f(r)\ket{\mathrm{pk}, \mathrm{Enc}(\mathrm{pk}, 1; r)}\right\}_\lambda
\end{equation}
where $f$ is the density function of the randomness in the $\mathrm{Enc}$ algorithm.
We also note that the above states are can be prepared efficiently (we say that they are \emph{BQP-preparable}) since they are just a classical mix and running the encryption algorithm with uniform coins results in the correct distribution.

We are now ready to prove our impossibility for zero-knowledge for mQMA.

\begin{theorem}[Zero-Knowledge for mQMA]
    If the LWE problem is hard, then there exist problems in mQMA without an efficient zero-knowledge protocol.
\end{theorem}
\begin{proof}
Consider the mixed state
    \begin{align*}
\rho^{(\mathrm{pk}_0,\mathrm{pk}_1, b_0, b_1)} &= \left(\sum_{r} f(r)\ket{\mathrm{pk}, \mathrm{Enc}(\mathrm{pk}, b_0; r)}_{C_0}\right)\otimes \left(\sum_{r} f(r)\ket{\mathrm{pk}, \mathrm{Enc}(\mathrm{pk}, b_1; r)}_{C_1}\right)
    \end{align*}
Then we define $X = (X^{\text{yes}}, X^{\text{no}})$, where
\[
X^{\text{yes}}=\left\{\rho^{(\mathrm{pk}_0, \mathrm{pk}_1, b_0, b_1)}\right\}_{\lambda,\mathrm{pk}_0, \mathrm{pk}_1, b_0, b_1: \{0\} \in \{b_0, b_1\}}
\quad
\text{and}
\quad
X^{\text{no}}=\left\{\rho^{(\mathrm{pk}_0, \mathrm{pk}_1, 1, 1)}\right\}_{\lambda,\mathrm{pk}_0, \mathrm{pk}_1}.
\]
We argue that $X =(X^{\text{yes}}, X^{\text{no}})$ is a valid mQMA instance. Indeed, for any given state $\rho \in X^{\text{yes}}$ a verifier, given the appropriate secret key $\mathrm{sk}$ as the witness, can efficiently text membership as follows. For $b\in\{0,1\}$:
\begin{itemize}
    \item Check if $\mathrm{pk}_b$ is a valid commitment for $\mathrm{sk}$.
    \item Use $\mathrm{sk}$ to decrypt the ciphertext in the $C_b$ register and check if the result is $0$.
\end{itemize}
Accept if either test passes.

Both tests are efficient and, by the perfect completeness of the encryption scheme, they succeed with certainty if the witness is valid. As for soundness, let $\rho \in X^{\text{no}}$. No (even unbounded) prover can convince a verifier about the validity of a key different from $\mathrm{sk}_0$ or $\mathrm{sk}_1$, by the statistically binding property of the commitment scheme. Thus, soundness follows by the decryption correctness of the encryption scheme.

Next, we show that such a language cannot have an efficient zero-knowledge protocol. 
Towards contradiction, fix any efficient zero-knowledge protocol for $X =(X^{\text{yes}}, X^{\text{no}})$ and let us assume without loss of generality that it has a negligible completeness and soundness error. This means that for a randomly sampled pair $\mathrm{pk}_0, \mathrm{pk}_1$ supported on $\mathrm{Gen}(1^\lambda)$ we have
\[
p_{000} = \Pr\left(\mathrm{P}\left(\mathrm{sk}_0,(\rho^{(\mathrm{pk}_0,\mathrm{pk}_1,0,0)})^{\otimes p(\lambda)} \right) \rightleftharpoons \mathrm{V}\left((\rho^{(\mathrm{pk}_0,\mathrm{pk}_1,0,0)})^{\otimes v(\lambda)}\right) = 1\right) \geq 1 - \mu(\lambda)
\]
for some negligible function $\mu$. Then let us consider what happens when we change the state input by the verifier to the quantity
\[
p_{010} = \Pr\left(\mathrm{P}\left(\mathrm{sk}_0,(\rho^{(\mathrm{pk}_0,\mathrm{pk}_1,0,0)})^{\otimes p(\lambda)} \right) \rightleftharpoons \mathrm{V}\left((\rho^{(\mathrm{pk}_0,\mathrm{pk}_1,1,0)})^{\otimes v(\lambda)}\right) = 1\right).
\]
We consider two complementary options:
\begin{itemize}
    \item There exists a negligible function $\mu$ such that $\abs{p_{000} - p_{010}} \leq \mu(\lambda)$.

    In this case, we claim that
    \begin{equation}\label{eq:arg1}
    \abs{p_{010} - \Pr\left(\mathrm{P}\left(\mathrm{sk}_0,(\rho^{(\mathrm{pk}_0,\mathrm{pk}_1,0,0)})^{\otimes p(\lambda)} \right) \rightleftharpoons \mathrm{V}\left((\rho^{(\mathrm{pk}_0,\mathrm{pk}_1,1,1)})^{\otimes v(\lambda)}\right) = 1\right)} \leq \mu(\lambda)
    \end{equation}
    for some negligible function $\mu$, and we defer the proof of \cref{eq:arg1} to a later point in this argument. However notice that this means that the verifier accepts with overwhelming probability a state $\rho^{(\mathrm{pk}_0,\mathrm{pk}_1,1,1)} \in X^{\text{no}}$, contradicting the soundness of the protocol.

    \item There exists a polynomial $p$ such that $\abs{p_{000} - p_{010}} \geq \frac{1}{p(\lambda)}$.

    For this case it is useful to recall that, again by correctness, switching the witness of the prover does not change the acceptance probability, so
    \[
p_{100} = \Pr\left(\mathrm{P}\left(\mathrm{sk}_1,(\rho^{(\mathrm{pk}_0,\mathrm{pk}_1,0,0)})^{\otimes p(\lambda)} \right) \rightleftharpoons \mathrm{V}\left((\rho^{(\mathrm{pk}_0,\mathrm{pk}_1,0,0)})^{\otimes v(\lambda)}\right) = 1\right) \geq 1 - \mu(\lambda)
\]
for some negligible function $\mu$. And furthermore we claim that
\begin{equation}\label{eq:arg2}
    \abs{p_{100} - \Pr\left(\mathrm{P}\left(\mathrm{sk}_1,(\rho^{(\mathrm{pk}_0,\mathrm{pk}_1,0,0)})^{\otimes p(\lambda)} \right) \rightleftharpoons \mathrm{V}\left((\rho^{(\mathrm{pk}_0,\mathrm{pk}_1,1,0)})^{\otimes v(\lambda)}\right) = 1\right)} \leq \mu(\lambda)
    \end{equation}
    and we defer the proof of \cref{eq:arg2} to a later point in this argument. However notice that there is now a non-negligible gap between the success probability of the prover using $\mathrm{sk}_0$ and the prover using $\mathrm{sk}_1$, if the verifier uses as its input $\rho^{(\mathrm{pk}_0,\mathrm{pk}_1,1,0)}$. This contradicts the witness indistinguishability (and therefore the zero-knowledge) property of the protocol.
\end{itemize}
Thus, we attain either a contradiction to soundness or zero-knowledge and all is left to be shown is that \cref{eq:arg1} and \cref{eq:arg2} hold. We show that \cref{eq:arg1} indeed holds and the same argument applies to \cref{eq:arg2}.

Note that the only difference between the two cases is whether the verifier inputs 
copies of $\rho^{(\mathrm{pk}_0,\mathrm{pk}_1,1,0)}$ or $\rho^{(\mathrm{pk}_0,\mathrm{pk}_1,1,1)}$. Note that both states are separable, BQP-preparable, and they are identical in the register $C_0$ so it suffices to consider the state in $C_1$: In the former case it consists of encryptions of $0$ under $\mathrm{pk}_1$ and in the latter it consist of encryptions of $1$ under $\mathrm{pk}_1$. Importantly, the prover is given $\mathrm{sk}_0$ as the witness and no additional information about $\mathrm{sk}_1$ is given in the protocol. Thus we can appeal to \cref{eq:CPA} (which is a consequence of CPA-security) to conclude that the two states
\[
\left\{\left(\mathrm{sk}_0, \rho^{(\mathrm{pk}_0,\mathrm{pk}_1,1,0)}\right) \right\}_\lambda \approx_c \left\{\left(\mathrm{sk}_0, \rho^{(\mathrm{pk}_0,\mathrm{pk}_1,1,1)}\right)\right\}_\lambda
\]
are computationally indistinguishable. The same holds for polynomially-many copies (by a standard hybrid argument), which in particular means that the success probability of the interactive protocol can only change by a negligible amount, since it is an efficiently observable event. Thus \cref{eq:arg1} holds.
\end{proof}

\fi
\section{Zero-Knowledge for pQIP}

In the following we present our main result, a zero-knowledge proof for any language in pQIP. We construct the protocol iteratively:
\begin{itemize}
    \item In \cref{sec:hvzk} we present a protocol that satisfies honest-verifier zero-knowledge.
    \item In \cref{sec:round} we show how to round-collapse the protocol to three messages.
    \item In \cref{sec:ampl} we discuss soundness amplification.
    \item In \cref{sec:public} we show how to make the protocol public-coin.
    \item In \cref{sec:zk} we finally show how to achieve security against a dishonest verifier.
\end{itemize}
Combining everything together, we obtain that every language in pQIP has a zero-knowledge protocol.
\begin{theorem}[EFIs Imply pQIPzk]
    If EFI pairs exist, then pQIP $\subseteq$ pQIPzk.
\end{theorem}

\subsection{Honest-Verifier Zero-Knowledge}\label{sec:hvzk}

Let $X = (X^{\text{yes}}, X^{\text{no}})$ be a pQIP problem instance, and we assume without loss of generality that the 
completeness-soundness gap is $1- \mu(\lambda)$, for some negligible function $\mu$ (see \cref{sec:def}).  We assume the existence of the following ingredients:
\begin{itemize}
    \item A one-sided statistically secure computation protocol (see \cref{sec:2pc}).
    \item A statistically hiding and computationally hiding commitment (see \cref{sec:qscom}).
\end{itemize}
Both of these ingredients can be constructed assuming the existence of EFI pairs. For a (pure) state $\ket{\psi} \in (\mathbb{C}^{2})^{\otimes n}$ and a security parameter $\lambda$, let $m=m(\lambda)$ be the number of messages of the pQIP protocol, which we assume to be odd for notational convenience. The protocol proceeds as follows.

\paragraph{Protocol I.} In the first round, the prover and the verifier engage in the secure computation of the following functionality, with statistical security against a corrupted prover.
\begin{itemize}
\item (Prover Input) The prover inputs nothing.
\item (Verifier Input) The verifier inputs copies of the state $\ket{\psi}^{\otimes v(\lambda)}$.
\item (Computation) The functionality proceeds as follows:
\begin{itemize}
\item Initialize a register $W$ with $\ketbra{0\dots 0}{0\dots 0} \otimes \ketbra{\psi}{\psi}^{\otimes v(\lambda)}$ and a register $E$ with $\ketbra*{0^\lambda}{0^\lambda}$.
\item Apply $\mathrm{Com}$ to $WE$ and divide the resulting state in two registers $C$ and $D$.
\item Return $C$ to the prover and $D$ to the verifier.
\end{itemize}
\end{itemize}
Let $r = \frac{m+1}{2}$ be the number of rounds of interaction. For $i \in \{1, \dots, r\}$ the prover and the verifier engage in the secure computation of the following functionality, with statistical security against a corrupted prover.
\begin{itemize}
    \item (Prover Input) The prover inputs a register $M$ and a register $C$.
    \item (Verifier Input) The verifier inputs a register $D$.
    \item (Computation) The functionality proceeds as follows:
    \begin{itemize}
    \item Apply the unitary $\mathrm{Com}^\dagger$ to $CD$, resulting in two registers $WE$.
    \item Apply the verification measurement $\Pi_{\mathrm{Com}}$ to $WE$ and return $\bot$ to both prover and verifier if the verification fails.
    \item Apply the verifier unitary $\mathrm{V}_i$ to $WM$.
    \item If $i=1$ measure the first qubit of $W$ in the computational basis and return its output to the verifier (and nothing to the prover).
    \item Otherwise, apply $\mathrm{Com}$ to $WE$ and divide the resulting state in two registers $C$ and $D$.
    \item Return $C$ and $M$ to the prover and $D$ to the verifier.
    \end{itemize}
\end{itemize}
If the secure computation protocol never aborts, the verifier returns the output of the last round of secure computation, otherwise it outputs $0$.

\paragraph{Analysis.} It is clear that the protocol is complete, since the prover can simply run the honest prover's strategy and cause the verifier to accept with overwhelming probability. Next, we show that the protocol is sound.

\begin{lemma}[Soundness]\label{lmm:hvzksound}
    Protocol I is sound.
\end{lemma}
\begin{proof}
Take an arbitrary prover $\mathcal{P}$ and consider the following series of hybrid experiments.
\begin{itemize}
    \item Hybrid $0$: This is the original interaction between the prover and the verifier.
    \item Hybrid $1$: We substitute all the interactions of the secure computation protocol with a the simulator. 
    
    Indistinguishability follows from the statistical security of the protocol and a triangle inequality.
    \item Hybrid $2$: We substitute the way we compute the commitments in the secure computation protocol. Before committing to $WE$, we initialize an external register $W'$ initialized to $\ketbra{0 \dots 0}{0 \dots 0}$ and we apply a $\mathrm{SWAP}$ operator to $WW'$. Then, before applying the verifier's unitary $\mathrm{V}_i$, we apply again the $\mathrm{SWAP}$ operator to $WW'$.

    Indistinguishability follows from a standard hybrid argument, where we reduce the security of each neighbouring hybrid to the (statistical) double opening binding property of the quantum state commitment scheme.
\end{itemize}
The proof is concluded by noticing that the last hybrid is identical to an execution of the pQIP protocol, where the verifier is run honestly. Thus soundness follows by a triangle inequality.
\end{proof}
Next, we define the notion of honest-verifier zero-knowledge.
\begin{definition}[Honest-Verifier Zero-Knowledge]\label{def:hvzk}
Let $X = (X^{\text{yes}}, X^{\text{no}})$ be a pQIP problem instance. An interactive protocol satisfies honest-verifier zero-knowledge for all 
$\ket{\psi} \in X^{\text{yes}}$ there exists a polynomials $s,v$ and a QPT simulator $\mathcal{S} = \{\mathcal{S}_\lambda\}_\lambda$ such that
\[
\left\{ \mathcal{S}_\lambda(i, \ket{\psi}^{\otimes s(\lambda)}) \right\}_\lambda \approx_c \left\{\Tr_R\left(\mathrm{P}_i\mathrm{V}_{i-1}\mathrm{P}_{i-1} \cdots \mathrm{V}_1\mathrm{P}_1 (\ket{\psi}^{\otimes v(\lambda)}\otimes \ket{0 \dots 0}) \right)\right\}_\lambda
\]
holds for all $i\in m(\lambda)$.
\end{definition}
We are now ready to prove that the protocol satisfies \cref{def:hvzk}.
\begin{lemma}[Honest-Verifier Zero-Knowledge]
    If EFI pairs exist, Protocol I is honest-verifier zero-knowledge.
\end{lemma}
\begin{proof}
    Fix any index $i$ and let $k \leq r$ be such that the $i$-the message of the protocol is a (possibly intermediate) message of the $k$-th secure computation sub-protocol. Consider the following sequence of hybrids.

\begin{itemize}
    \item Hybrid $0$: The output of the simulator is defined to be
    \[
    \Tr_R\left(\mathrm{P}_i\mathrm{V}_{i-1}\mathrm{P}_{i-1} \cdots \mathrm{V}_1\mathrm{P}_1 (\ket{\psi}^{\otimes v(\lambda)}\otimes \ket{0 \dots 0})  \right)
    \]
    that is, it corresponds to the real experiment.

    \item Hybrid $1$: We invoke the simulator of the secure computation protocol to simulate the messages of the $k$-th computation protocol.

    Indistinguishability follows by a reduction to the (computational) malicious security of the protocol: The reduction is given as a non-uniform advice the states in registers $C$, $D$, and $M$ corresponding to an honest execution of the protocol, right before the beginning of the $k$-th secure computation. The reduction uses these registers as the inputs for its interaction with the challenger.  
    Then, it uses the distinguisher to detect if the interaction was real or simulated.

    \item Hybrid $2$: If $k = r$, simulate the output of the secure computation protocol to be $1$, otherwise simulate the output to commit to the $\ketbra{0\dots 0}{0\dots 0}$ state.

    For the former case, indistinguishability follows by the negligible completeness error of the protocol. For the latter, by a non-uniform reduction to the (computational) double opening binding property of the quantum state commitment scheme. 
\end{itemize}
Then, for $j \in \{1, \dots, k-1\}$ we define
\begin{itemize}
    \item Hybrid $2j + 1$: The simulator invokes the simulator of the secure computation protocol to simulate the messages of the $(k-j)$-th secure computation protocol.

    Indistinguishability follows again by a non-uniform reduction, as above. However, note that messages corresponding to subsequent sub-protocols are all efficient (since subsequent protocols are simulated) and so the reduction can efficiently complete the execution of the protocol to get to the final state.
    
    \item Hybrid $2j + 2$: The simulator programs the output of the $(k-j)$-th secure computation protocol to a commitment to the $\ketbra{0\dots 0}{0\dots 0}$ state.

    Indistinguishability follows from the same argument as above.
\end{itemize}
Finally, observe that in the last hybrid the simulator computes the final state by simulating instances of the secure computation protocol and programming their output to be a commitment to a fixed state, which is also efficiently computable. Overall, this establishes that the simulator is QPT and its output is computationally indistinguishable from the original message.
\end{proof}

\subsection{Round Collapsing}\label{sec:round}

Let $(\mathrm{P}, \mathrm{V})$ be an interactive protocol for pQIP that satisfies honest-verifier zero-knowledge (see \cref{sec:hvzk}) and has negligible completeness and soundness errors. Let us denote by 
\[
\ketbra{\psi_v}{\psi_v} = \ketbra{\psi}{\psi}^{\otimes v(\lambda)} \otimes \ketbra{0\dots 0}{0\dots 0}
\]
the initial state of a the verifier, for a problem instance $\ket{\psi} \in (\mathbb{C}^{2})^{\otimes n}$. In the following we show that a small modification of the \cite{KW00} compiler allows us to turn any such $r$-round protocol into a $3$-message one.

\paragraph{Protocol II.} We describe the protocol from the perspective of the verifier  in the following.
\begin{itemize}
    \item Receive registers $(W_2, \dots, W_r)$ and $(M_1, \dots, M_r)$ from the prover. Initialize a register $W_1$ to $\ketbra{\psi_v}{\psi_v}$. Apply the projective measurement $\Pi_{\mathrm{Acc}} = \{ \Pi_0, \Pi_1\}$ to $W_rM_r$, where we define
    \[
    \Pi_0 = \mathrm{Id}- \Pi_1
    \quad \text{and} \quad
    \Pi_1 = \mathrm{V}_r^\dagger (\ketbra{1}{1} \otimes \mathrm{Id}) \mathrm{V}_r.
    \]
    Reject if the outcome $1$ is not observed.
    \item Initialize two registers $BB'$ with the maximally entangled state
    \[
    \frac{\ket{0}_B\ket{0}_{B'} + \ket{1}_B\ket{1}_{B'}}{\sqrt{2}}
    \]
    and sample a random $i \in \{1, \dots, r-1\}$. Apply $\mathrm{V}_i$ to $W_iM_i$ then apply a $\mathrm{SWAP}$ operator between $W_i$ and $W_{i+1}$, controlled on the register $B$. Send the index $i$ and the registers $M_i$, $M_{i+1}$, and $B'$ to the prover.
    \item Receive the register $B'$ back from the prover. Controlled on $B$ apply the Pauli $\mathrm{X}$ to $B'$, then measure $B$ in the Hadamard basis. Accept of the outcome $+$ is observed, and reject otherwise. 
\end{itemize}

\paragraph{Analysis.} For an honest prover, the protocol is identical to \cite{KW00}, so completeness follows from the same argument. For soundness, we show that the same argument from \cite{KW00} also applies to our modified protocol. Let us denote the projection $\Pi_{\mathrm{Init}} = \ketbra{\psi_v}{\psi_v}$ and recall the following technical lemma.
\begin{lemma}[\cite{KW00}]\label{lmm:5}
Let $(\mathrm{P}, \mathrm{V})$ be an interactive protocol with $r$ rounds and soundness $\zeta$. Let $\rho_1, \dots \rho_r$ be states in registers $M_iW_i$ satisfying
\[
\rho_1 = \Pi_{\mathrm{Init}}\rho_1\Pi_{\mathrm{Init}}
\quad\text{and}\quad
\rho_r = \Pi_1\rho_r \Pi_1.
\]
Then \[\sum_{i=1}^{r-1}\sqrt{F\left(\Tr_{M_i}(\mathrm{V}_i^\dagger\rho_i \mathrm{V}_i), Tr_{M_i}(\rho_{i+1})\right)} \leq r-1-\frac{(1-\zeta)^2}{8(r-1)}.\]
\end{lemma}

We are now ready to prove soundness.
\begin{lemma}[Soundness]
If $(\mathrm{P}, \mathrm{V})$ satisfies soundness with error $\zeta$, then Protocol II is sound with error $1 - \frac{(1-\zeta)^2}{16(r-1)^2}$.
\end{lemma}
\begin{proof}
    We bound the probability of the prover by the probability that the state \emph{post-selected} on passing the check of projective measurement $\Pi_{\mathrm{Acc}}$. Given that this is a necessary condition for the prover to pass, this can only increase its success probability. 
    
    Let $\rho_1, \dots, \rho_r$ be the states held be the verifier in registers $M_iW_i$, after the aforementioned post-selection. We want to derive a closed form for the probability of the verifier accepting for a given $i$, denoted by $p_i$. The state of the verifier after the Pauli $\mathrm{X}$ applied in the third round can be written
    \[
    \frac{\ket{0}_B \ket{\psi_0}_{W_iK} + \ket{1}_B \ket{\psi_1}_{W_iK}}{\sqrt{2}}
    \]
    where $K$ denotes the collection of all registers but $B$ and $W_i$. Measuring this in the Hadamard basis, we obtain $0$ with probability
    \begin{equation}
    \label{eq:pi}
    p_i =  \norm{\frac{1}{2}\ket{0} (\ket{\psi_0}+ \ket{\psi_1})}^2 = \frac{1}{2} + \frac{1}{2}\Re\ip{\psi_0}{\psi_1}.
    \end{equation}
    Note that 
    \[
    \Tr_K(\ketbra{\psi_0}{\psi_0}) = \Tr_{M_i}(\mathrm{V}_i^\dagger\rho_i \mathrm{V}_i)
    \quad\text{and}\quad
    \Tr_K(\ketbra{\psi_1}{\psi_1}) = \Tr_{M_i}(\rho_{i+1})
    \]
    thus by \cref{eq:fidelity} we can bound
\begin{equation}\label{eq:abs}
     \abs{\ip{\psi_0}{\psi_1}}^2 \leq F\left(\Tr_{M_i}(\mathrm{V}_i^\dagger\rho_i \mathrm{V}_i), Tr_{M_i}(\rho_{i+1})\right).
\end{equation}
Combining \cref{eq:pi} and \cref{eq:abs}, we can conclude that 
\[
p_i \leq \frac{1}{2} + \frac{1}{2}\sqrt{F\left(\Tr_{M_i}(\mathrm{V}_i^\dagger\rho_i \mathrm{V}_i), Tr_{M_i}(\rho_{i+1})\right)}.
\]
We can now appeal to \cref{lmm:5} to bound
\[
\E_i (p_i) \leq 1 - \frac{(1-\zeta)^2}{16(r-1)^2}
\]
since the register $W_1$ is initialized to $\ketbra{\psi_v}{\psi_v}$ and it is therefore trivially in the image of $\Pi_{\mathrm{Init}}$, whereas the condition on $\rho_r$ holds by assumption. This concludes the proof.
\end{proof}

We show that the protocol satisfies honest-verifier zero-knowledge, following an argument from \cite{KobayashiZK}.
\begin{lemma}[Honest-Verifier Zero-Knowledge]
If $(\mathrm{P}, \mathrm{V})$ satisfies honest-verifier zero-knowledge, then so does Protocol II.
\end{lemma}
\begin{proof}
Given that the interactive protocol satisfies completeness with probability $1-\mu(\lambda)$, for a negligible function $\mu$, it suffices to describe a QPT simulator that produces an accepting transcript. Let $\mathcal{S}$ be the simulator of the underlying protocol and let $S$ be its work register, which we assume to be initialized to $\ket{0 \dots 0}$. For $k \in \{1, \dots, r-1\}$ the simulator prepares the registers $M_kW_k$ by computing
\[
\Tr_S\left(\mathcal{S}_k \mathrm{V}_{k-1} \dots \mathrm{V}_1 \mathcal{S}_1 (\ket{\psi_v}\ket{0 \dots 0})\right)
\]
which the simulator can compute with sufficiently many copies of $\ket{\psi}$. Given that the transcript is accepting and that the verifier is honest, the measurement $\Pi_{\mathrm{Acc}}$ acts as the identity on the state sent by the simulator.

Upon receiving the index $i$ and the registers $M_iM_{i+1}B'$, the simulator acts as follows:
\begin{itemize}
\item Apply $\mathcal{S}_{i+1}$ to $M_iS_i$, where $S_i$ is the simulator internal register corresponding to the $i$-th copy above.
\item Controlled on $B'$, apply $\mathrm{SWAP}$ to $M_iS_i$ with $M_{i+1}S_{i+1}$. 
\end{itemize}
To see why this is a good simulation, let us denote the joint state of the registers $S_kW_kM_k$ right after the simulator's first message by
\[
\ket{\psi_k}= \mathcal{S}_k \mathrm{V}_{k-1} \dots \mathrm{V}_1 \mathcal{S}_1 (\ket{\psi_v}\ket{0 \dots 0})
\]
and we have that
\[
\mathrm{V}_k \ket{\psi_k} = \mathrm{V}_k \mathcal{S}_k \mathrm{V}_{k-1} \dots \mathrm{V}_1 \mathcal{S}_1 (\ket{\psi_v}\ket{0 \dots 0}) = \mathcal{S}_{k+1}^\dagger \ket{\psi_{k+1}}.
\]
Then, the joint state between the prover and the verifier after the verifier's message is
\begin{align*}
&\ket{0}_{B}\ket{0}_{B'}\mathrm{V}_{i}\ket{\psi_{i}}_{S_iW_iM_i}\ket{\psi_{i+1}}_{S_{i+1}W_{i+1}M_{i+1}}\\ & \quad +
\ket{1}_{B}\ket{1}_{B'}\mathrm{SWAP}_{W_iW_{i+1}}\mathrm{V}_{i}\ket{\psi_{i}}_{S_iW_iM_i}\ket{\psi_{i+1}}_{S_{i+1}W_{i+1}M_{i+1}}\\
=~&\ket{0}_{B}\ket{0}_{B'}\mathcal{S}_{i+1}^\dagger\ket{\psi_{i+1}}_{S_iW_iM_i}\ket{\psi_{i+1}}_{S_{i+1}W_{i+1}M_{i+1}}\\  & \quad +\ket{1}_{B}\ket{1}_{B'}\mathrm{SWAP}_{W_iW_{i+1}}\mathcal{S}_{i+1}^\dagger\ket{\psi_{i+1}}_{S_iW_iM_i}\ket{\psi_{i+1}}_{S_{i+1}W_{i+1}M_{i+1}}
\end{align*}
omitting normalization factors. Then, after the unitary $\mathcal{S}_{i+1}$ applied by the simulator we obtain the (unnormalized) state
\ifllncs
\begin{align*}
&\ket{0}_{B}\ket{0}_{B'}\ket{\psi_{i+1}}_{S_iW_iM_i}\ket{\psi_{i+1}}_{S_{i+1}W_{i+1}M_{i+1}}\\
&+\ket{1}_{B}\ket{1}_{B'}\mathrm{SWAP}_{W_iW_{i+1}}\ket{\psi_{i+1}}_{S_iW_iM_i}\ket{\psi_{i+1}}_{S_{i+1}W_{i+1}M_{i+1}}
\end{align*}
\else
\[
\ket{0}_{B}\ket{0}_{B'}\ket{\psi_{i+1}}_{S_iW_iM_i}\ket{\psi_{i+1}}_{S_{i+1}W_{i+1}M_{i+1}}+\ket{1}_{B}\ket{1}_{B'}\mathrm{SWAP}_{W_iW_{i+1}}\ket{\psi_{i+1}}_{S_iW_iM_i}\ket{\psi_{i+1}}_{S_{i+1}W_{i+1}M_{i+1}}
\]
\fi
and the controlled $\mathrm{SWAP}$ operator results into the identity 
\ifllncs
\begin{align*}
&\mathrm{SWAP}_{M_iS_iM_{i+1}S_{i+1}}
\mathrm{SWAP}_{W_iW_{i+1}}\ket{\psi_{i+1}}_{S_iW_iM_i}\ket{\psi_{i+1}}_{S_{i+1}W_{i+1}M_{i+1}}\\=&\ket{\psi_{i+1}}_{S_iW_iM_i}\ket{\psi_{i+1}}_{S_{i+1}W_{i+1}M_{i+1}}
\end{align*}
\else
\[
\mathrm{SWAP}_{M_iS_iM_{i+1}S_{i+1}}
\mathrm{SWAP}_{W_iW_{i+1}}\ket{\psi_{i+1}}_{S_iW_iM_i}\ket{\psi_{i+1}}_{S_{i+1}W_{i+1}M_{i+1}}=\ket{\psi_{i+1}}_{S_iW_iM_i}\ket{\psi_{i+1}}_{S_{i+1}W_{i+1}M_{i+1}}
\]
\fi
which allows us to factorize the state as
\[
\left(\ket{0}_{B}\ket{0}_{B'}+\ket{1}_{B}\ket{1}_{B'}\right)\otimes \ket{\psi_{i+1}}_{S_iW_iM_i}\ket{\psi_{i+1}}_{S_{i+1}W_{i+1}M_{i+1}}.
\]
Note that the final check of the verifier is a projection of the $BB'$ registers onto the maximally entangled state, which succeeds with certainty on the above state. This concludes the proof.
\end{proof}

\subsection{Amplification}\label{sec:ampl}

It is shown in \cite{KW00} that the soundness of any $3$-message protocol amplifies with parallel repetition. Technically, \cite{KW00} only claimed the soundness result for interactive protocols where the initial state of the verifier is $\ket{0\dots 0}$, whereas here we need a slightly generalized version, where the initial state of the verifier is arbitrary. However, their proof works perfectly well for this case as well.

One way to see this is that we can substitute our verifier with a (possibly inefficient) one that starts with the $\ket{0\dots 0}$ state and applies the unitary $\mathrm{U}_v$ such that $\mathrm{U}_v\ket{0\dots 0} = \ket{\psi_v}$ and proceed as specified in the protocol. The proof of the parallel repetition theorem does not use the fact that the verifier is efficient, so the statement follows. 

\begin{lemma}[\cite{KW00}]
Let $(\mathrm{P}, \mathrm{V})$ be a $3$-message protocol with soundness error $\zeta$. Then, for any polynomial $p$, its $p$-fold parallel repetition has soundness error $\zeta^p$. 
\end{lemma}
Clearly honest-verifier zero-knowledge is preserved under parallel repetition (by a standard hybrid argument). Hence, we can henceforth assume that our protocol has negligible soundness error.

\subsection{Public Coin}\label{sec:public}

Let $(\mathrm{P}, \mathrm{V})$ be the parallel repeated version of the $3$-message protocol from \cref{sec:round} so that it has negligible completeness and soundness errors. Since the protocol has exactly $3$ messages, it is fully specified by an initial state plus two prover unitaries $(\mathrm{P}_1, \mathrm{P}_2)$ and two verifier unitaries $(\mathrm{V}_1, \mathrm{V}_2)$.

In the following we show how to compile the protocol into one where the verifier sends only a single bit, by slightly modifying the compiler from \cite{MW05}. To this end, let the problem instance be $\ket{\psi} \in (\mathbb{C}^{2})^{\otimes n}$ and note that the initial state of the verifier $\ket{\psi_v}$ is efficiently preparable given enough copies of $\ket{\psi}$. 

\paragraph{Protocol III.} The compiled protocol proceeds as follows.
\begin{itemize}
    \item The verifier receives a register $W$ from the prover.
    \item The verifier samples a coin $b \in\{0,1\}$ uniformly at random and sends it to the prover.
    \item The prover returns a register $M$.
    \begin{itemize}
    \item If $b=0$: Apply $\mathrm{V}_2$ to $WM$ and measure the first qubit of $W$ in the computational basis. Accept if the outcome is $1$.
    \item If $b=1$: Apply $\mathrm{V}_1^\dagger$ to $WM$ and perform a SWAP-test between $\ket{\psi_v}$ and $W$. Accept if the SWAP-test accepts.
    \end{itemize}
\end{itemize}

\paragraph{Analysis.} Completeness follows immediately from the one of the underlying protocol. Next, we show that the protocol is sound. Before that, we need to recall two technical lemmas from \cite{MW05}.

\begin{lemma}[\cite{MW05}]\label{lmm:mw1}
    Let $\rho_{WM}$ be a mixed state and let $\rho_W = \Tr_M(\rho_{WM})$. If $MW$ is measured with respect to a binary-valued projective measurement $\{\Pi_0, \Pi_1\}$, then then the probability of obtaining the outcome 1 is at most $F(\sigma_W, \rho_W)^2$ for some $\sigma_W$ a reduced density of a state in the support of $\Pi_1$. 
\end{lemma}
\begin{lemma}[\cite{MW05}]\label{lmm:mw2}
    The maximum probability with which a verifier specified by $\mathrm{V}_1$ and $\mathrm{V}_2$ can be made to accept is
    \[
  \max_{\sigma, \tau}  F(\sigma, \tau)^2
    \]
    where the maximum is taken over $\sigma$ in the support of $\mathrm{V}_2^\dagger (\ketbra{1}{1}\otimes \mathrm{Id})\mathrm{V}_2$ and $\tau$ in the support of $\mathrm{V}_1\ketbra{\psi_v}{\psi_v}\mathrm{V}_1^\dagger$.
\end{lemma}
We are now ready to show the soundness of our protocol.
 \begin{lemma}[Soundness]
If $(\mathrm{P}, \mathrm{V})$ satisfies soundness with error $\zeta$, then Protocol II is sound with error $\frac{3}{4} + \frac{\sqrt{\zeta}}{2}$.
\end{lemma}

\begin{proof}
    Let $\rho \in \mathcal{H}_W$ be the state received in the first round by the verifier, and let $\rho_0, \rho_1 \in \mathcal{H}_W \otimes \mathcal{H}_M$ be the states after the second message of the prover, conditioned on $b=0$ and $b=1$, respectively. By \cref{eq:swapPOVM}, the success probability of the above protocol is
    \begin{align*}
        &\frac{1}{2}\Tr\left(\mathrm{V}_2^\dagger (\ketbra{1}{1}\otimes \mathrm{Id})\mathrm{V}_2 \rho_0\right) + 
        \frac{1}{2}\Tr\left(\frac{\mathrm{Id} + \ketbra{\psi_v}{\psi_v}}{2} \mathrm{V}_1^\dagger\rho_1\mathrm{V}_1\right) \\
        =~& \frac{1}{2}\Tr\left(\mathrm{V}_2^\dagger (\ketbra{1}{1}\otimes \mathrm{Id})\mathrm{V}_2 \rho_0\right) + \frac{1}{4}\Tr\left(\mathrm{V}_1^\dagger\rho_1\mathrm{V}_1\right) +
        \frac{1}{4}\Tr\left(\mathrm{V}_1\ketbra{\psi_v}{\psi_v}\mathrm{V}_1^\dagger \rho_1\right)\\
        \leq~& \frac{1}{2}\Tr\left(\mathrm{V}_2^\dagger (\ketbra{1}{1}\otimes \mathrm{Id})\mathrm{V}_2 \rho_0\right) + \frac{1}{4} +
        \frac{1}{4}\Tr\left(\mathrm{V}_1\ketbra{\psi_v}{\psi_v}\mathrm{V}_1^\dagger \rho_1\right)\\
        \leq~& \frac{1}{4} +\underbrace{\frac{1}{2}\Tr\left(\mathrm{V}_2^\dagger (\ketbra{1}{1}\otimes \mathrm{Id})\mathrm{V}_2 \rho_0\right) + 
        \frac{1}{2}\Tr\left(\mathrm{V}_1\ketbra{\psi_v}{\psi_v}\mathrm{V}_1^\dagger \rho_1\right)}_{=: p}
    \end{align*}
    where $p$ is the success probability of an identical protocol, except that instead of performing a SWAP-test we are projecting onto $\ketbra{\psi_v}{\psi_v}$. Thus, it suffices to bound $p$.

    By \cref{lmm:mw1} we have that 
    \[
    p \leq \max_{\sigma, \zeta} \left(\frac{1}{2}F(\sigma, \rho)^2 +\frac{1}{2}F(\rho, \tau)^2\right)
    \]
    maximized over $\sigma$ in the support of $\mathrm{V}_2^\dagger (\ketbra{1}{1}\otimes \mathrm{Id})\mathrm{V}_2$ and $\tau$ in the support of $\mathrm{V}_1\ketbra{\psi_v}{\psi_v}\mathrm{V}_1^\dagger$. For all $\sigma, \tau$, we can bound this expression as
    \[
    \frac{1}{2}F(\sigma, \rho)^2 +\frac{1}{2}F(\rho, \tau)^2 \leq \frac{1}{2} + F(\sigma, \tau) \leq \frac{1}{2} + \frac{\sqrt{\zeta}}{2}
    \]
    where the first inequality follows by \cref{eq:fidtriangle} and the second inequality follows by \cref{lmm:mw2}.
\end{proof}
Finally, we show that the protocol satisfies honest-verifier zero-knowledge, adapting an argument from \cite{KobayashiZK}.
\begin{lemma}[Honest-Verifier Zero-Knowledge]
If $(\mathrm{P}, \mathrm{V})$ satisfies honest-verifier zero-knowledge, then so does Protocol III.
\end{lemma}
\begin{proof}
    Given that the completeness error is negligible, it suffices to show a simulator that produces an accepting transcript, given a simulator $\mathcal{S}$ for the input $3$-message protocol. The simulator works as follows:
    \begin{itemize}
    \item With probability $1/2$, 
    set the registers $MW$ to the state
    \[
    \Tr_S\left(\mathcal{S}_2 \mathrm{V}_1 \mathcal{S}_1 (\ket{\psi_v}\ket{0 \dots 0})\right)
    \]
    and return these registers along with $b=0$.
    \item With probability $1/2$, set the registers $MW$ to the state
    \[
    \Tr_S\left(\mathrm{V}_1 \mathcal{S}_1 (\ket{\psi_v}\ket{0 \dots 0})\right)
    \]
    and return these registers along with $b=1$.
\end{itemize}
It is easy to see that the simulator's output is indistinguishable from an honest accepting transcript, and furthermore it is efficiently computatable.
\end{proof}

\subsection{Zero-Knowledge}\label{sec:zk}

Let $(\mathrm{P}, \mathrm{V})$ be the $3$-message protocol presented in \cref{sec:public}, where the verifier sends single bits. We show how to compile it to guarantee zero-knowledge against a \emph{malicious verifiers}. The compilation is simple: Protocol IV is defined to be identical to Protocol III, except that instead of the verifier sampling a random bit $b$, the parties engage in the secure computation of the following functionality, with statistical security against a corrupted prover.
\begin{itemize}
\item (Prover Input) The prover inputs a random bit $b_p$
\item (Verifier Input) The verifier inputs a random bit $b_v$
\item (Computation) The functionality computes $b_p \oplus b_v$ and outputs it to both parties.
\end{itemize}
To amplify soundness, then the protocol is repeated \emph{sequentially.}

Completeness is immediate. Soundness is also a straightforward reduction against the security of the secure computation protocol. The lemma below proves constant soundness and it is known that sequential repetition amplifies soundness (unconditionally).
\begin{lemma}[Soundness]
If $(\mathrm{P}, \mathrm{V})$ satisfies soundness with error $\zeta$, then one iteration Protocol IV is sound with error negligibly close to $\zeta$.
\end{lemma}
\begin{proof}
    Consider the hybrid world where the simulator is interacting with the prover to run the secure computation protocol. By \cref{def:2pc}, the two experiments are statistically close, thus the acceptance probability can only change by a negligible amount. However, note that in the simulated experiment the simulator extracts the bit of the prover $b_p$ and then afterwards samples a random $b_v$ to return $b_p \oplus b_v$. Thus, the output bit is an unbiased random bit and the soundness follows by that of the underlying protocol.
\end{proof}
Finally, we prove that the scheme satisfies regular zero-knowledge.
\begin{lemma}[Zero-Knowledge]
    If EFI pairs exist and $(\mathrm{P}, \mathrm{V})$ satisfies honest-verifier zero-knowledge, Protocol IV is zero-knowledge.
\end{lemma}
\begin{proof}
    Let $\ell$ be the number of sequential repetition of the above protocol. We define the simulator gradually in a sequence of hybrids. For $i \in \{1, \dots, \ell\}$ we define
    \begin{itemize}
        \item Hybrid $i$: We run the simulator for the honest-verifier zero-knowledge protocol, which returns two registers $WM$ and a bit $b$. For the $(\ell -i+1)$-th iteration of the protocol we send the register $W$ as the first message, then we simulate the secure computation protocol so that the output equals $b$ and we return $M$ as the last message.
    \end{itemize}
    In the last hybrid, the simulation is runs in QPT, so all is left to be shown is that the neighbouring hybrids are indistinguishable to any QPT verifier. 
    
    Assume towards contradiction that there exists an $i$ such that the hybrids $i$ and $i+1$ are distinguishable with some inverse-polynomial advantage. As a thought experiment, consider a modified version of the hybrids as follows:
    \begin{itemize}
        \item Hybrid $i.0$: Same as the $i$-th hybrid.
        \item Hybrid $i.1$: Here the experiment samples a bit $b^*\in\{0,1\}$ and aborts if the output of the $(\ell -i)$-th secure computation protocol does not equal $b^*$.

        Clearly, the experiment aborts with probability exactly $1/2$, independently of the random variables in the experiment. Thus the distinguishing advantage between hybrids $i.1$ and $(i+1).1$ is at most half the advantage between $i$ and $i+1$.

        \item Hybrid $i.2$: Here the experiment sets $b^*$ to be the bit output by the honest-verifier simulator, and uses the corresponding registers $MW$ as the first and last message in the $(\ell -i)$-th iteration.

        Note the the distribution of $b^*$ is identical to the previous hybrid, so the only difference is that the registers of the honest-verifier simulation are used in place of the real one. Indistinguishability follows from a non-uniform reduction against the honest-verifier zero-knowledge property, where the reduction is given as advice the state of the system right before the  $(\ell -i)$-th iteration (subsequent iterations are efficiently simulatable by assumption).

        \item Hybrid $i.3$: Here the experiment simulates the secure computation protocol with a randomly sampled bit as output.

        Indistinguishability follows from another non-uniform reduction against the malicious security of the secure computation protocol. Once again, the reduction is given the state of the system right before the  $(\ell -i)$-th iteration as a non-uniform advice. Subsequent iterations are, by assumption, efficiently simulatable.
\end{itemize}
     The proof is concluded by observing that hybrid $i.3$ and $(i+1).1$ are identical, which contradicts the fact $i$ and $i+1$ (and consequently $i.1$ and $(i+1).1$) are distinguishable with inverse-polynomial advantage.
\end{proof}

\ifllncs
\else
\section{Zero-Knowledge for the Uhlmann Transformation}

The Uhlmann transformation problem is specified by two P-uniform unitaries $C,D$ that, when applied to the $\ket{0}$ state yield
\[
C\ket{0} = \ket{C} \in \mathcal{H}_{R}\otimes \mathcal{H}_{S}
\quad \text{and}\quad
D\ket{0} = \ket{D} \in \mathcal{H}_{R}\otimes \mathcal{H}_{S}
\]
such that $\mathrm{Tr}_S(\ket{C}) = \mathrm{Tr}_S(\ket{D})$. By Uhlmann's theorem \cite{UHLMANN} there exists an unitary\footnote{In general $\mathrm{U}$ is only unitary if the reduced densities of $\ketbra{C}{C}$ and $\ketbra{D}{D}$ are invertible, and otherwise it is a partial isometry, i.e., the restriction of a unitary to a subspace. For simplicity here we only consider the unitary case and we refer the reader to \cite{BEMPQY23} for a more general treatment.} $\mathrm{U}$ acting on the $S$ register such that
\[
(\mathrm{Id}\otimes \mathrm{U})\ket{C} = \ket{D}.
\]
The distributional Uhlmann transformation problem (with perfect fidelity\footnote{Generalizations with imperfect fidelity are also possible, although we do not explore them here.}) \cite{BEMPQY23} requires one to build an interactive protocol where the verifier holds a register $T$ and the prover is supposed to apply $\mathrm{U}$ to $T$. In \cite{BEMPQY23} the authors present a protocol to prove this in zero-knowledge against an \emph{honest verifier} and here we show how to modify their protocol to achieve zero-knowledge also against a malicious verifier.

\paragraph{Secure Computation for Reactive Functionalities.} A \emph{reactive} functionality \cite{IPS08,EFI} is a functionality where one of the two parties (without loss of generality, Bob) receives intermediate outputs and it is allowed to respond with an intermediate input, which is added to the current state of the computation. The security definition is identical to \cref{def:2pc}, except that the simulator has multiple rounds of interaction with the ideal functionality, each time extracting an input from the adversary.

The protocol presented in \cref{sec:hvzk} achieves precisely this, except that the verifier is given as output a single bit. It is straightforward to modify the protocol so that the output of the verifier is an arbitrary state. Then, the same argument that we used in \cref{sec:hvzk} can be used to show that the compiler results into a secure computation protocol for a reactive (for the prover) functionality against malicious parties. However, contrary to \cref{sec:hvzk}, we can only prove this for the case where the interactive protocol is \emph{efficient}.

\begin{lemma}[Reactive Secure Computation]\label{lmm:reactive}
    If EFI pairs exist, then there exist a reactive secure computation protocol for all efficient interactive protocols $(\mathrm{Alice},\mathrm{Bob})$.
\end{lemma}
\begin{proof}[Proof Sketch]
    Statistical security against a corrupted Bob is identical to \cref{lmm:hvzksound} and it is omitted. On the other hand, security against a corrupted Alice follows by a series of hybrids:
    \begin{itemize}
    \item Hybrid $0$: This is the original interaction between Alice and Bob.
    \item Hybrid $1$: We substitute all the interactions of the secure computation protocol with a the simulator. 

    Since the interaction between Alice and Bob is efficient, we can simulate it, while reducing against the \emph{computational} security of the secure computation protocol.
    
    \item Hybrid $2$: We substitute the way we compute the commitments in the secure computation protocol. Instead of committing to the current state of the workspace of Alice, we commit to the $\ketbra{0 \dots 0}{0 \dots 0}$ state.

    Indistinguishability follows from the computational double opening security of the commitment scheme.
\end{itemize}
In the last hybrid, the output of the computation is obtained by running an honest interaction between Alice and Bob and the simulator programs it as the output for Alice.
\end{proof}

\paragraph{Zero-Knowledge Uhlmann Beyond Honest Verifier.} Armed with our a secure computation for reactive functionalities (with one-sided) statistical security against Bob, we are now ready to describe our protocol against a malicious verifier. Alas, \cref{lmm:reactive} only holds for efficient prover/verifier protocols, so a straightforward composition with the protocol from \cite{BEMPQY23} does not work. However, we notice that the the number of qubits of the quantum state $n$ also imposes a bound on the runtime of the honest prover, since it can be easily simulated by a quantum circuit with polynomial in $2^n$ number of gates \cite{NC00}. Thus, we can resort to a technique known as \emph{complexity leveraging}, where we set the security parameter to, say,  $\lambda = n^{10}$ and, now assuming \emph{sub-exponential security} of EFI pairs, we can run the same argument as \cref{lmm:reactive}.

We show how to compose the protocol from \cite{BEMPQY23} with our secure computation for reactive functionalities. Let $\delta = \delta(\lambda)$ be a parameter to be set later, let $\gamma = 8\delta^2$, and let $n$ be a bound on the number of qubits of the states. The protocol consists of a run of the secure computation (with statistical security against the prover) with parameter $\lambda = n^c$, for some constant $c$, for the following reactive functionality.
\begin{itemize}
    \item (Prover Input) The prover inputs a classical description of the circuits $C$ and $D$.
    \item (Verifier Input) The verifier inputs a classical description of the circuits $C$ and $D$. along with a target register $T$.
    \item (Computation) The reactive functionality proceeds as follows:
    \begin{itemize}
    \item Abort if the description of the circuits $C$ and $D$ don't match.
    \item Sample a random index $i^* \in \{1, \dots, \gamma\}$.
    \item Then for all $i \in  \{1, \dots, \gamma\}$, proceed as follows.
    \item If $i \neq i^*$:
    \begin{itemize}
        \item Run the $C$ circuit to produce the state $\ketbra{C}{C}_{R,S}$.
        \item Send the register $R$ to the prover.
        \item Upon receiving the register $R$ back, apply the projective measurement $\Pi_D = \{ \Pi_0, \Pi_1\}$ where
        \[
        \Pi_0 = \mathrm{Id} - \Pi_1 \quad\text{and}\quad \Pi_1 = \ketbra{D}{D}
        \]
        which is efficiently implementable given $D$. Abort if the output $1$ is not observed.
    \end{itemize}        
    \item If $i = i^*$:
    \begin{itemize}
        \item Send the $T$ register to the prover.
        \item Receive the register $T$ back.
    \end{itemize}     
    \item If none of the tests failed, return $T$.
\end{itemize}
\end{itemize}
The runtime of the verifier is polynomial in $\lambda = n^{O(1)}$, which implies that the verifier is QPT (even though the prover might not be). Furthermore, note that the ideal functionality as described above is precisely identical to the protocol described in \cite{BEMPQY23}, except for the choice of the register $T$, which is left up to the verifier. As shown in \cite{BEMPQY23}, the probability that an honest prover and an honest verifier successfully terminate (i.e., the protocol does not abort), is $1$. 

A straightforward reduction shows that the protocol is (unconditionally) sound.

\begin{lemma}[Soundness]
    Let $(C,D)$ be an Uhlmann problem instance with error parameter $\delta$ and unitary $\mathrm{U}$. For sufficiently large $n$, for all provers $\mathcal{P}$ there exists a negligible function $\mu$ such that 
    \[
    \text{if }\Pr\left(\mathrm{V}\left(C,D,\Tr_S(\ket{C}_{T,S})\right)\rightleftharpoons\mathcal{P}(C,D)\right) \geq \frac{1}{2} \quad \text{then } \mathrm{Td}(\sigma, (\mathrm{U} \otimes \mathrm{Id})\ketbra{C}{C}) \leq \frac{1}{\delta} + \mu(\lambda)
    \]
    where $\sigma$ is the output of the protocol, conditioned on not aborting.
\end{lemma}
\begin{proof}[Proof Sketch]
    By the (statistical) simulation security of the reactive secure computation protocol against a corrupted prover, we can equivalently analyze the interaction between the prover and the ideal functionality, at the cost of a negligible term in the trace distance. Then soundness follows from Lemma 6.3 in \cite{BEMPQY23}.
\end{proof}
Next, we define the zero-knowledge property against a malicious verifier. Even at a definitional level, it is not immediately clear what is the right generalization of the notion of \cite{BEMPQY23}, since the choice of the register $T$ is left completely up to the verifier. If the verifier is honest, it will automatically input the reduced density of $\ketbra{C}{C}$, but otherwise we have no guarantee. One sensible choice of definition is that we can simulate the interaction with a corrupted verifier in the presence of an oracle that the simulator can query exactly \emph{once}. This models the fact that all a corrupted verifier can learn is the output of $\mathrm{U}$ on a single state of its choice, but nothing more. 

\begin{lemma}[Zero-Knowledge]
    If sub-exponentially secure EFI pairs exists, then the following holds. For any (malicious) QPT verifier $\mathcal{V} = \{\mathcal{V}_\lambda\}_\lambda$, there exists a QPT simulator $\mathcal{S} = \{\mathcal{S}_\lambda\}_\lambda$, such that for any QPT distinguisher $\mathcal{D} = \{\mathcal{D}_\lambda\}_\lambda$, there exists a negligible function $\mu$ such that for all security parameters $\lambda\in\mathbb{N}$, all non-uniform bipartite advice state $\rho_{VA}$, and all Uhlmann instances $(C,D)$ with unitary $\mathrm{U}$ it holds that:
    \[
    \abs{\Pr\left(\mathcal{D}_\lambda\left( \mathrm{P}_\lambda(C,D) \rightleftharpoons \mathcal{V}(\rho_V,C,D),\rho_A\right)= 1\right) -
    \Pr\left(\mathcal{D}_\lambda\left( \mathcal{S}_{\lambda,\mathrm{U}}(1^\lambda, \rho_V, C, D),\rho_A\right)= 1\right)
    } \leq \mu(\lambda)
    \]
    where the algorithm $\mathcal{S}_{\lambda,\mathrm{U}}$ is defined as follows:
    \begin{itemize}
    \item Run the simulator $\mathcal{S}_\lambda(1^\lambda, \rho_V)$, which outputs a register $T$.
    \item Apply $\mathrm{U}$ to $T$.
    \item Continue executing the simulator with $T$, which in the end outputs a certain state $\sigma$.
    \item Output $\sigma$.
    \end{itemize}
\end{lemma}

\begin{proof}[Proof Sketch]
    We run the simulator from \cref{lmm:reactive} to simulate the secure computation protocol against a malicious verifier. The runtime of the simulator is polynomial in $2^n$, which still suffices to derive a contradiction to the sub-exponential security of EFI pairs. Then, we can conclude that the output of the simulator is computationally indistinguishable from an honest run.

    However, the simulator is not QPT, since it needs to run an honest interaction between the prover and the verifier in order to compute the output state in register $T$ (and note that this is the only inefficient computation that the simulator makes). To make it QPT, we instead extract $T$ from the inputs of the (malicious) verifier and use the oracle provided by the experiment to apply $\mathrm{U}$ to $T$. By assumption, this is identical to an honest run, and the simulator is now QPT.
\end{proof}

\subsection*{Acknowledgments}

We thank Kai-Min Chung for comments on an earlier draft of this work. 

Research supported by the European Research Council through an ERC Starting Grant (Grant agreement No.~101077455, ObfusQation) and partially funded by the Deutsche Forschungsgemeinschaft (DFG, German Research Foundation) under Germany's Excellence Strategy - EXC 2092 CASA – 390781972.
\fi

\ifllncs
\newpage
\else
\fi

\bibliographystyle{alpha}
\bibliography{references}

\ifllncs
\appendix

\else
\fi

\end{document}